%% file: main.tex
\documentclass{article}
\usepackage{PRIMEarxiv}

\usepackage[utf8]{inputenc} 
\usepackage[T1]{fontenc}    
\usepackage{hyperref}       
\usepackage{url}            
\usepackage{booktabs}       
\usepackage{amsfonts}       
\usepackage{nicefrac}       
\usepackage{microtype}      
\usepackage{lipsum}
\usepackage{fancyhdr}       
\usepackage{graphicx}       
\graphicspath{{media/}}     

\usepackage{xcolor}
\usepackage{graphicx}
\usepackage{dcolumn}
\usepackage{bm}
\usepackage{amsmath}
\usepackage{amssymb}
\usepackage{mathtools}
\usepackage{amsthm}
\usepackage{multirow}
\usepackage{rotating}
\usepackage{algorithm}
\usepackage{algpseudocode}
\usepackage{multirow}
 
\newtheorem{theorem}{Theorem}

\title{Uncertainty Quantification in Working Memory via Moment Neural Networks}
\author{Hengyuan Ma}
\date{September 2024}

\def\mm{\textit{Methods}}

\def\si{\textit{Supplementary Information}}

\begin{document}

\begin{flushleft}

{\Large
Uncertainty Quantification in Working Memory via Moment Neural Networks
}
\newline
{
  Hengyuan Ma\textsuperscript{1},
  Wenlian Lu\textsuperscript{1,2,3,4,5,6},
  Jianfeng Feng\textsuperscript{1,2,3,4,7$\ast$} 
\\

\bigskip
\it{1} Institute of Science and Technology for Brain-inspired Intelligence, Fudan University, Shanghai 200433, China
\\
\it{2} Key Laboratory of Computational Neuroscience and Brain-Inspired Intelligence (Fudan University), Ministry of Education, China
\\
\it{3} School of Mathematical Sciences, Fudan University, No. 220 Handan Road, Shanghai, 200433, Shanghai, China
\\
\it{4} Shanghai Center for Mathematical Sciences, No. 220 Handan Road, Shanghai, 200433, Shanghai, China
\\
\it{5} Shanghai Key Laboratory for Contemporary Applied Mathematics, No. 220 Handan Road, Shanghai, 200433, Shanghai, China
\\
\it{6} Key Laboratory of Mathematics for Nonlinear Science, No. 220 Handan Road, Shanghai, 200433, Shanghai, China 
\\
\it{7} Department of Computer Science, University of Warwick, Coventry, CV4 7AL, UK
\newline
$\ast$ jffeng@fudan.edu.cn 
}
\end{flushleft}

\input{files/0_abstract}
\input{files/1_intro}
\input{files/2_task}
\input{files/3_model}
\input{files/4_training}
\input{files/5_results}

\input{files/6_mechanism}
\input{files/7_snn}
\input{files/8_parameter}
\input{files/9_discussion}
\input{files/10_methods}

\bibliographystyle{unsrt}  
\bibliography{reference}  

\clearpage

\input{files/Supplementary}

\end{document}

%% file: files/0_abstract.tex
\begin{abstract}
Humans possess a finely tuned sense of uncertainty that helps anticipate potential errors, vital for adaptive behavior and survival. However, the underlying neural mechanisms remain unclear. This study applies moment neural networks (MNNs) to explore the neural mechanism of uncertainty quantification in working memory (WM). The MNN captures nonlinear coupling of the first two moments in spiking neural networks (SNNs), identifying firing covariance as a key indicator of uncertainty in encoded information. Trained on a WM task, the model demonstrates coding precision and uncertainty quantification comparable to human performance. Analysis reveals a link between the probabilistic and sampling-based coding for uncertainty representation. Transferring the MNN's weights to an SNN replicates these results. Furthermore, the study provides testable predictions demonstrating how noise and heterogeneity enhance WM performance, highlighting their beneficial role rather than being mere biological byproducts. These findings offer insights into how the brain effectively manages uncertainty with exceptional accuracy.
\end{abstract}

%% file: files/1_intro.tex
\section{Introduction}

Humans not only make decisions but also assess the degree of confidence associated with those decisions. This ability is crucial because it enables individuals to adapt their behavior based on the reliability of their judgments, enabling more effective navigation of uncertain environments. For instance, Accurately assessing confidence (or quantifying uncertainty) allows humans to prioritize cognitive resources, seek targeted information to resolve specific uncertainties, and effectively communicate their confidence levels to guide collaborative decision-making. Importantly, extensive studies have shown that humans possess the ability to generate a sense of uncertainty that accurately reflects the likelihood of errors across various tasks~\cite{keshvari2012probabilistic,devkar2017monkeys,honig2020humans,jabar2020using,yoo2021uncertainty}. In other words, humans tend to feel less confident when they make larger mistakes, and more confident when their errors are smaller. This capability gives humans an advantage over artificial intelligence in tasks requiring trustworthiness and reliability, as deep networks often become overconfident in their predictions even when large mistakes occur~\cite{gawlikowski2023survey}. However, the neural mechanisms underlying this precise uncertainty representation remain unclear. Elucidating the neural basis of uncertainty representation is a fundamental challenge in neuroscience~\cite{walker2023studying}. Advancing this understanding also has implications for designing more trustworthy and interpretable artificial intelligence systems~\cite{seuss2021bridging}. Notably, research in this domain has extensively utilized second-order moments for uncertainty quantification (UQ)~\cite{gast2018lightweight,postels2019sampling}. This raises an intriguing question: do humans similarly rely on second-order moments of neural activities for UQ?

There has been a long-standing debate on how neuron population activity represents uncertainty~\cite{ma2014neural}. Two leading theories, probabilistic population coding~\cite{ma2006bayesian} and sampling-based coding~\cite{orban2016neural}, offer different interpretations. Probabilistic population coding suggests that both the mean of the target variable and its uncertainty are captured by time-averaged neural responses. In contrast, sampling-based coding attributes the mean representation to the average response and uncertainty to response variability. While both frameworks successfully link neural activity to behavioral variability observed in experiments, neither fully captures the precise representation of uncertainty required to account for observed errors~\cite{honig2020humans,jabar2020using,yoo2021uncertainty,li2021joint}.

Working memory (WM) is the basis of numerous high-level cognitive processes~\cite{bays2024representation}. Understanding the neural mechanisms underlying uncertainty representation and quantification in WM is essential for explaining how humans manage uncertainty and generate confidence in various high-level tasks. Although WM is prone to errors due to both internal and external noise~\cite{oberauer2018benchmarks}, humans are aware of the potential errors through a sense of uncertainty~\cite{li2021joint} and leverage this awareness to optimize their performance~\cite{honig2020humans}. Additionally, significant fluctuations in WM error levels have been observed~\cite{van2012variability,keshvari2012probabilistic,rademaker2012introspective}, highlighting the importance of reliable UQ for recognizing deficiencies in WM content. 
Although probabilistic models of UQ have been proposed~\cite{henaff2020representation,li2021joint}, its neural mechanisms and a biologically plausible implementations, such as spiking neural networks (SNN), remain largely underexplored.

Ring attractor neural networks are prominent models for WM, supported by experimental evidence~\cite{wimmer2014bump,kim2017ring,chaudhuri2019intrinsic,khona2022attractor}. These models apply a ring manifold to store a continuous feature, such as head direction~\cite{kutschireiter2023bayesian}. However, they rely on highly structured synaptic connections~\cite{burak2009accurate,wu2016continuous}, which contradict the observed heterogeneity in neural tuning functions and synaptic connections~\cite{fisher2013modeling,chaudhuri2019intrinsic}. Additionally, WM errors are attributed to diffusion of the bump location, which fails to explain how humans quantify uncertainty in WM, as the variance of the decoded feature remains constant despite variations in actual memory error~\cite{keshvari2012probabilistic,rademaker2012introspective}.

Instead of manual design, there is a growing trend toward training recurrent neural networks for cognitive tasks~\cite{mante2013context,song2016training}, which facilitates the discovery of diverse network configurations. However, many of these studies~\cite{orhan2019diverse} rely on backpropagation through time for training~\cite{song2016training}, which is not biologically plausible. Recently, Darshan et al. (2022) proposed to train networks with synaptic heterogeneity using reservoir computing for WM~\cite{darshan2022learning}. While this approach is backpropagation-free and breaks the structured network connection, it introduces systematic drift, leading to significant coding errors without a mechanism to capture the associated uncertainty. Additionally, like many studies~\cite{orhan2019diverse,echeveste2020cortical}, they use rate-based models, which fail to account for biologically realistic features, such as spike-based communication in neurons. Training more realistic models, like SNNs, remains much more challenging.

In this study, we address the above issues by employing moment neural networks (MNNs)~\cite{lu2010gaussian}. Unlike rate-based neural models, which simplify SNNs by considering only the mean firing rate (mean), MNNs capture the complex nonlinear dynamics arising from both the mean firing rate and firing covariance (covariance) of leaky integrate-and-fire SNNs. Moreover, MNNs use differentiable moment activations (Eq.~\eqref{eq:moment_activation}, \mm) instead of discrete spikes, allowing for easier training. MNNs have proven to be a valuable model for studying the neural basis of various cognitive functions. They have revealed how covariance influences coding precision and memory capacity in working memory (WM)~\cite{ma2023self}, demonstrated how covariance encodes or extracts information during perceptual tasks~\cite{zhu2024learning}, and shown how spiking neurons facilitate faster decision-making processes~\cite{zhu2024towards}.

Our results show that, trained using the reservoir computing approach introduced in~\cite{darshan2022learning}, MNNs effectively capture uncertainty in WM through covariance, achieving coding precision comparable to human performance in WM tasks, even with significant synaptic heterogeneity. We then propose a potential mechanism for uncertainty representation in WM, revealing a direct link between probabilistic population coding and sampling-based coding through the nonlinear coupling of mean and covariance. These findings, derived from the MNN, are validated on the SNN by transferring the trained MNN weights. This not only confirms that the MNN serves as a faithful substitute for the SNN but also achieves the first implementation of uncertainty quantification in WM via spike-based communications. Additionally, we present testable predictions on how factors such as noise, neural correlations, and heterogeneity influence WM performance, highlighting their beneficial roles during network training and inference. Overall, this study offers new insights into how humans manage uncertainty with remarkable precision through correlated neuronal fluctuations.

%% file: files/2_task.tex
\input{fig/memory}

\section{Working memory and uncertainty quantification}

WM is one of the most widely studied cognitive functions, serving as the foundation for various cognitive processes, including learning and comprehension, attention and focus, goal-oriented behavior, and adaptability~\cite{oberauer2002access}.
Previous studies have shown that WM is susceptible to errors, leading to variations in precision that are reflected in behavioral mistakes~\cite{oberauer2018benchmarks}. However, humans possess an awareness of these potential errors, which manifests as a sense of confidence about precision of remembered content~\cite{li2021joint}. Further research indicates that individuals can leverage this awareness to optimize their performance~\cite{honig2020humans}.

As shown in Fig.~\ref{fig:memory}a, the WM task designed by Li et al.~\cite{li2021joint} involves participants first being presented with a cue that provides information about a feature (location) to be memorized. This is followed by a delay period during which the cue is removed. After the delay, participants are asked to report the remembered feature and indicate their uncertainty by adjusting the length of an arc centered at the reported location, representing a confidence interval for estimating the true position.
To encourage participants to quantify uncertainty about the target location accurately, the experiment awarded points based on their responses. Points were granted only if the target location fell within the designated arc, with the score decreasing as the arc length increased. To maximize their scores, participants were expected to use shorter arcs when they were more confident.
Broadly, responses can be categorized into four types, as illustrated in Fig.~\ref{fig:memory}b. Ideal uncertainty quantification (UQ) occurs when the arc length is large for large errors and small for small errors. The other two cases represent overconfidence and underconfidence. The study reported a strong correlation (correlation coefficient of 0.6–0.7) between arc length and error (Fig.~5D in~\cite{li2021joint}), indicating that participants’ reported uncertainty reliably reflected the error in their responses.

%% file: fig/memory.tex
\begin{figure}
    \centering
    \includegraphics[width=0.99\linewidth]{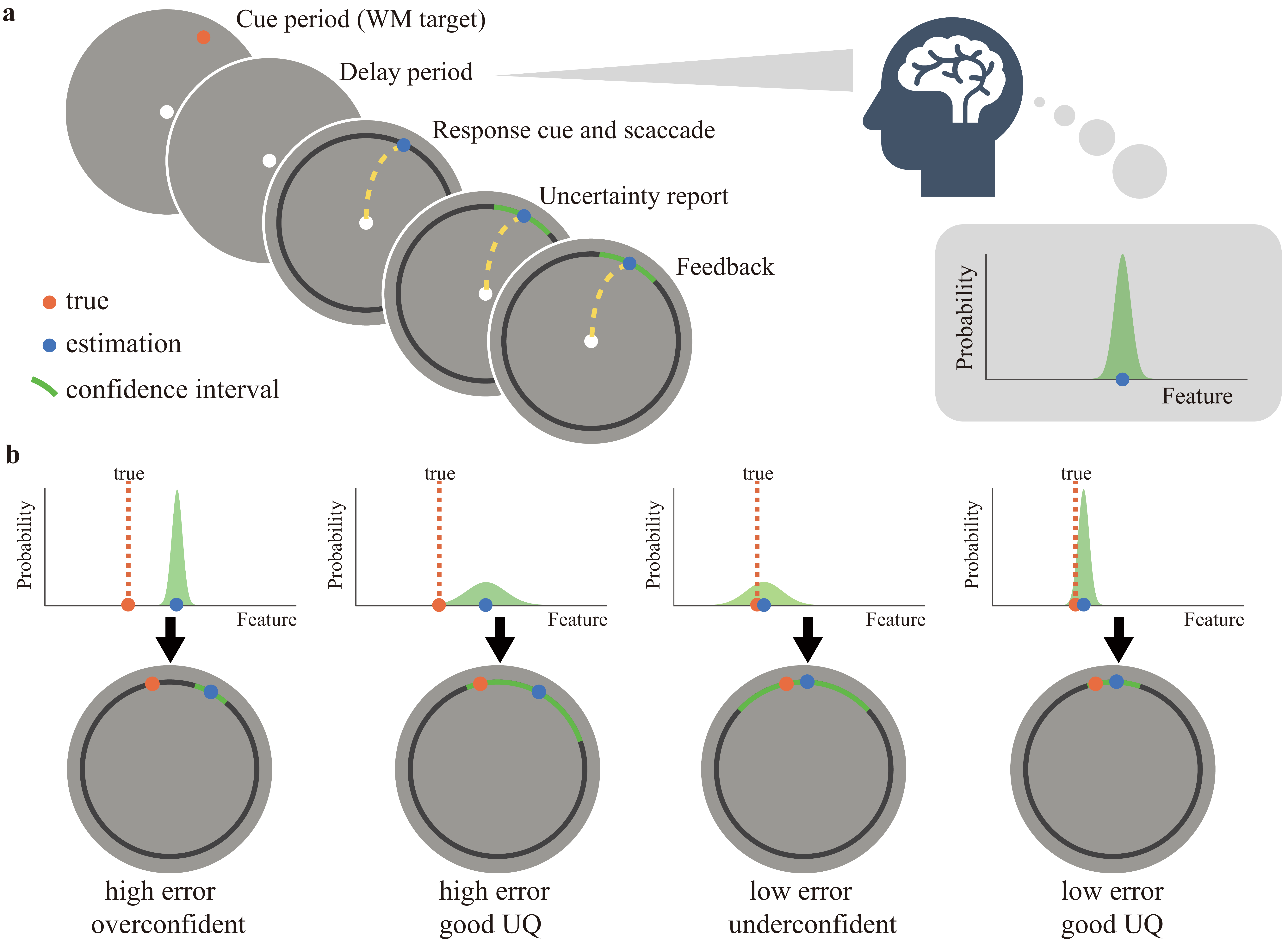}
    \caption{{\bf Working memory task and its uncertainty quantification (UQ). }(a) In the task designed in~\cite{li2021joint}, participants are required to remember the location indicated by the cue. After a delay period, they use saccades to indicate the remembered location on the ring and report their uncertainty with an arc.
    (b) Four representative cases of UQ results. Effective UQ should accurately reflect the magnitude of the error.}\label{fig:memory}
\end{figure}

%% file: files/3_model.tex
\section{Neural models for working memory}

\input{fig/spike_vs_rate}

\subsection{Spiking neural network}

Spiking neural networks (SNNs) (Eq.~\eqref{eq:LIF}, \mm) are widely used to model the dynamics of biological neural systems underlying various cognitive tasks, as illustrated in the top of Fig.~\ref{fig:spike_vs_rate}a. However, the discrete nature of spike trains presents challenges for model construction and analysis. To address this, spike trains can be simplified by extracting statistical features such as mean and covariance, as shown in Fig.~\ref{fig:spike_vs_rate}b. When only the mean is considered, spiking neural networks can be represented as rate-based neural networks.

\subsection{Rate-based neural network}
Rate-based neural models, which serve as a simplification of SNNs, are widely used to model WM. Its dynamic is defined as
\begin{align}
    \tau \frac{\partial \bm{\mu}}{\partial t}  = -\bm{\mu}+\phi(\bar{\bm{\mu}})\label{eq:rate} ,
\end{align}
where $\bm{\mu}\in\mathbb{R}^N$ represent the mean of the spike count per unit time (mean firing rate), $\phi(\cdot)$ is the element-wise nonlinearity, which is often set as Sigmoid or tanh function, $\tau$ is the membrane time constant, and $\bar{\bm{\mu}}\in\mathbb{R}^N$ corresponds to the mean of the input current for each neuron calculated as $\bar{\bm{\mu}} = W\bm{\mu}+\bm{\mu}_s$,
where $W\in\mathbb{R}^{N\times N}$ is the synaptic connection matrix, and $\bm{\mu}_s$ is the external input current.
We illustrate the rate-based neural network in the top of Fig~\ref{fig:spike_vs_rate}b.
To model WM, the weight connections in the rate-based neural model are set to be shift-symmetric, with long-distance inhibition and short-distance excitation. This configuration ensures that the system's fixed points are identical bumps that shift to form a ring manifold~\cite{wimmer2014bump,qi2015subdiffusive}. The location of the bump encodes the feature being remembered. To incorporate internal neuronal noise, white noise is exerted to the network, causing diffusion of the bump location, with variance increasing linearly over time\cite{burak2012fundamental}. However, it cannot account for the variability in precision observed in WM experiments conducted under identical conditions~\cite{fougnie2012variability,van2012variability}, hence inadequate for explaining the reliable UQ performance observed in humans.

We highlight a key distinction between spike-based biological neural systems and rate-based neural networks. In the former, each computational unit is governed by point-process (spikes), while in the latter, units are represented by scalar values. These scalars approximate the spike point process, representing its firing rate. However, this approximation overlooks the variability inherent in point processes and the nonlinear dynamics of such variability. Previous studies have demonstrated the significant role this variability plays in neural computation~\cite{churchland2011variance,polk2012correlated}.

\subsection{Moment neural networks}

A major challenge faced by rate-based neural models is difficulties in capturing realistic neuronal fluctuations (firing covariance) that are nonlinearly coupled with the the mean , which plays an important role in the neural dynamics~\cite{lu2010gaussian,helias2014correlation,dahmen2016correlated}.
To capture the nonlinear coupling of correlated fluctuations in the recurrent population of spiking neurons, we employ a model known as the moment neural network (MNN)~\cite{lu2010gaussian} defined as
\begin{align}\label{eq:mnn}
\begin{aligned}
\left\{\begin{matrix}
        &\tau \frac{\partial \bm{\mu}}{\partial t}  = -\bm{\mu}+\phi_{\mu}(\bar{\bm{\mu}},\bar{C}) \\ 
    &\tau \frac{\partial C}{\partial t}= -C+
 \phi_{C}(\bar{\bm{\mu}},\bar{C})
\end{matrix}\right.
\end{aligned}
\end{align}
where firing covariability $C\in\mathbb{R}^{N\times N}$ represent the covariance of the spike count per unit time, respectively. The covariance of the total synaptic current input $\bar{C}\in\mathbb{R}^{N\times N}$ is calculated as $\bar{C} = WCW^\top+\sigma^2_sI$,
where $\sigma_s$ is the noise level during the training phase.
The moment activations $\phi_{\mu}$ and $\phi_{C}$ (defined 
in Eq.~\eqref{eq:moment_activation}, \mm) describes the relationship between the input current statistics and the output spike train statistics in the leaky integrate-and fire (LIF) spiking neural model. 
Unlike rate-based neural models which typically use heuristic activation functions such as tanh or sigmoid~\cite{amari1977dynamics,coombes2005waves,coombes2010neural,wu2016continuous}, the nonlinearity of MNN are derived through a mathematical technique known as the diffusion approximation~\cite{amit1991quantitative,amit1997dynamics} which faithfully captures the nonlinear coupling of mean and firing variability across populations of spiking neurons.

We highlight three key aspects of the MNN. First, unlike rate-based neural models (top of Fig.~\ref{fig:spike_vs_rate}b), the MNN captures the nonlinear coupling between mean and covariance, as illustrated in the bottom of Fig.~\ref{fig:spike_vs_rate}b. Second, the pattern of the covariance of the input synaptic current to the neural population $\bar{C}$ that emerges from the model is a result of the intrinsic dynamics of the recurrent circuit, rather than being determined by external input, as in rate-based neural networks, where the input is uncorrelated.
Third, the nonlinear coupling between mean and covariance during MNN training allows the network to regulate both simultaneously in task performance. For instance, the error signal during training enables the model to adjust covariance to represent task-related uncertainty, a capability that networks based solely on mean cannot achieve. 

%% file: fig/spike_vs_rate.tex
\begin{figure}
    \centering
    \includegraphics[width=0.95\linewidth]{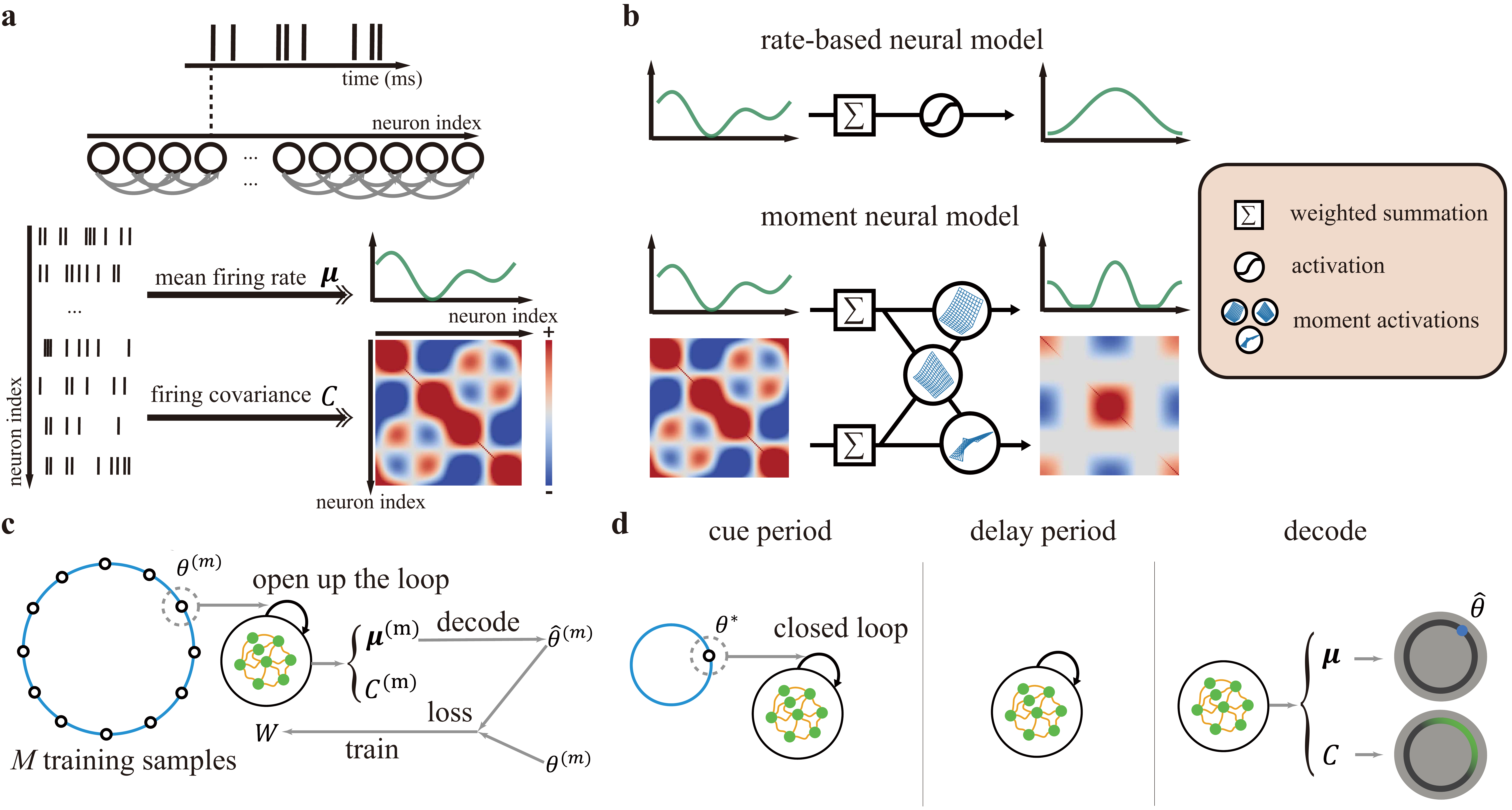}
    \caption{{\bf Comparison of the spiking  neural network (SNN), rate-based neural network, and moment neural network (MNN).} (a). (Top) An scheme of an SNN. 
    (Bottom) Spike trains of a neuron population can be summarized by the mean firing rate of neurons and the firing covariance between each pair of neurons.
    (b). (Top) The rate-based neural model only considers the mean firing rate and its nonlinear evolution through an activation function.
     (Bottom) The MNN captures both the mean firing rate and firing covariance, with their nonlinear coupling during network evolution represented by the moment activations.
    (c) Training the synaptic connection of an MNN for the working memory task using reservoir computing.  
    (d) The inference procedure is as follows: During the cue period, a feature encoded by external input is sent to the trained MNN. During the delay period, the external inputs are removed. After the delay period, the feature estimation and its uncertainty are decoded from the mean and covariance of the MNN, respectively. }
    \label{fig:spike_vs_rate}
\end{figure}

%% file: files/4_training.tex
\section{Training the MNN for working memory}

We trained an MNN to hold WM of a continuous periodic feature $\theta \in [0, 2\pi)$. During training, we selected $M$ training features $\theta^{(m)}, m = 1, \ldots, M$, each encoded as an external current $\bm{\mu}_s^{(m)}$. We trained the connections $W$ so that the mean of the fixed point encodes the value of $\theta^{(m)}$ under input $\bm{\mu}_s^{(m)}$.
We applied a reservoir computing approach following~\cite{darshan2022learning}, with $L_2$ regularization controlled by a factor $\alpha$ (Eq.~\eqref{eq:opt} in \mm), which encourages weaker or sparser connections, making the network more energy-efficient. See \mm~for details. The training procedure is illustrated in Fig.~\ref{fig:spike_vs_rate}c.

Two key points should be emphasized. First, similar to the approach in~\cite{darshan2022learning}, the connection matrix \( W \) includes an untrainable random matrix component \( J \in \mathbb{R}^{N \times N} \), which introduces unstructured heterogeneity, consistent with experimental findings that neuronal connection dynamics are highly heterogeneous~\cite{renart2003robust,murray2017stable,langdon2023unifying}. The degree of this heterogeneity is controlled by a factor \( g > 0 \). While some studies have shown that structured heterogeneity can reduce error levels in WM~\cite{kilpatrick2013optimizing}, unstructured heterogeneity has been shown to often lead to a rapid decline in the network's computational capabilities~\cite{burak2009accurate}. This highlights the challenge of maintaining reasonable cognitive precision of our model.  
Second, the training loss (Eq.~\eqref{eq:opt} in \mm) does not include the covariance \( C \). Consequently, only the mean is directly supervised during training, while the covariance remains unsupervised. This distinguishes our MNN from many models in machine learning for UQ, which typically require supervision of both the mean and the (co)variance of the model's output through loss functions such as log-likelihood or evidence lower bound~\cite{blundell2015weight,gast2018lightweight,wu2018deterministic}. In contrast, the covariance in our MNN adapts to its inputs indirectly, through the nonlinear coupling between the mean and the covariance.

After training, we tested the network for the WM task on $L$ instances, as illustrated in Fig.~\ref{fig:spike_vs_rate}d. We first input cue that convey the information of a target feature $\theta^{\ast}$ to the network during the cue period. Then, during the delay period, we removed the external current and introduced a new parameter $\sigma_{s,2}$ to control the noise level during the inference phase (see Eq.~\eqref{eq:inference_input}, \mm). After the delay period ends, we decoded the remembered feature \(\hat{\theta}\) and its associated uncertainty level $\kappa$ from the network state as \((\bm{\mu}, C)\). We quantified the uncertainty level represented by the covariance based on four uncertainty metrics I-IV (Eq.~\eqref{eq:uncern_1} and Eq.~\eqref{eq:uncern_2}, \mm), and we also record the estimation error $e$. An effective UQ is expected to show a strong positive correlation coefficient between the uncertainty $\kappa$ and the error $e$, which is estimated by Eq.~\eqref{eq:rho}, \mm.

%% file: files/5_results.tex
\input{fig/uq_result}

\section{Precise coding and reliable uncertainty quantification in the MNN}

We analyzed the trained MNN as follows. Unlike previous studies where the bump attractor network's connection matrix is homogeneously designed~\cite{burak2012fundamental,wimmer2014bump,ma2023self}, our MNN employs a connection matrix trained through reservoir computing, a biologically plausible algorithm~\cite{maass2002real}, rather than being manually specified. This approach makes our model more biologically realistic: both the tuning curves (Fig.~\ref{fig:uq_result}a) and the connections (Fig.~\ref{fig:uq_result}b) are heterogeneous, aligning with experimental findings~\cite{langdon2023unifying}. Interestingly, each neuron only responds to a narrow range of features, suggesting a sparse coding scheme often observed experimentally~\cite{beyeler2019neural}.

We then present the decoding results of the MNN at the inference phase. 
Six instances of the decoded results are shown in Fig.~\ref{fig:uq_result}c, illustrating that error levels vary significantly across different instances. This variability is consistent with experimental observations of substantial fluctuations in the quality of WM representations within an individual~\cite{fougnie2012variability,van2012variability}. More importantly, we observe that larger errors are associated with higher covariance levels, suggesting that covariance effectively captures uncertainty.
Furthermore, we divided all $L$ test instances into three groups based on the uncertainty level: the low-confidence group corresponds to the top 25\% of uncertainty indicator \(I\), the middle-confidence group corresponds to the 25\%-50\% range of uncertainty indicator \(I\), and the remaining instances correspond to the high-confidence group. As shown in Fig.\ref{fig:uq_result}d, the low-confidence group exhibits the heaviest-tailed error distribution, followed by the middle-confidence group. This observation aligns with the experimental results in~\cite{rademaker2012introspective,fougnie2012variability}. 
Additionally, the precision of our model also reaches comparable level reported in~\cite{rademaker2012introspective,fougnie2012variability,li2021joint}.

We then systematically evaluated the UQ capability of the MNN. As shown in Fig.~\ref{fig:uq_result}e, the uncertainty levels $\kappa$ calculated from the I-IV indicators are all strongly positively correlated with the error $e$, with correlation coefficients comparable to those reported in~\cite{li2021joint}. This suggests that the covariance in the MNN effectively captures uncertainty.

To examine the role of neuron firing correlation in UQ, we clamped all pairwise correlations between neurons in the covariance matrix to zero in the MNN and repeated the experiments as an ablation study. As shown in Fig.~\ref{fig:uq_result}e, we observed a significant decrease in the correlation coefficient $\rho$ between the error $e$ and each of the four indicators. This finding suggests that the correlation between neural activities is crucial for the neural system's ability to quantify uncertainty. It aligns with studies showing through animal experiments that correlated fluctuations between neurons are essential for population coding~\cite{berkes2011spontaneous}.

%% file: fig/uq_result.tex
\begin{figure}
    \centering
    \includegraphics[width=0.99\linewidth]{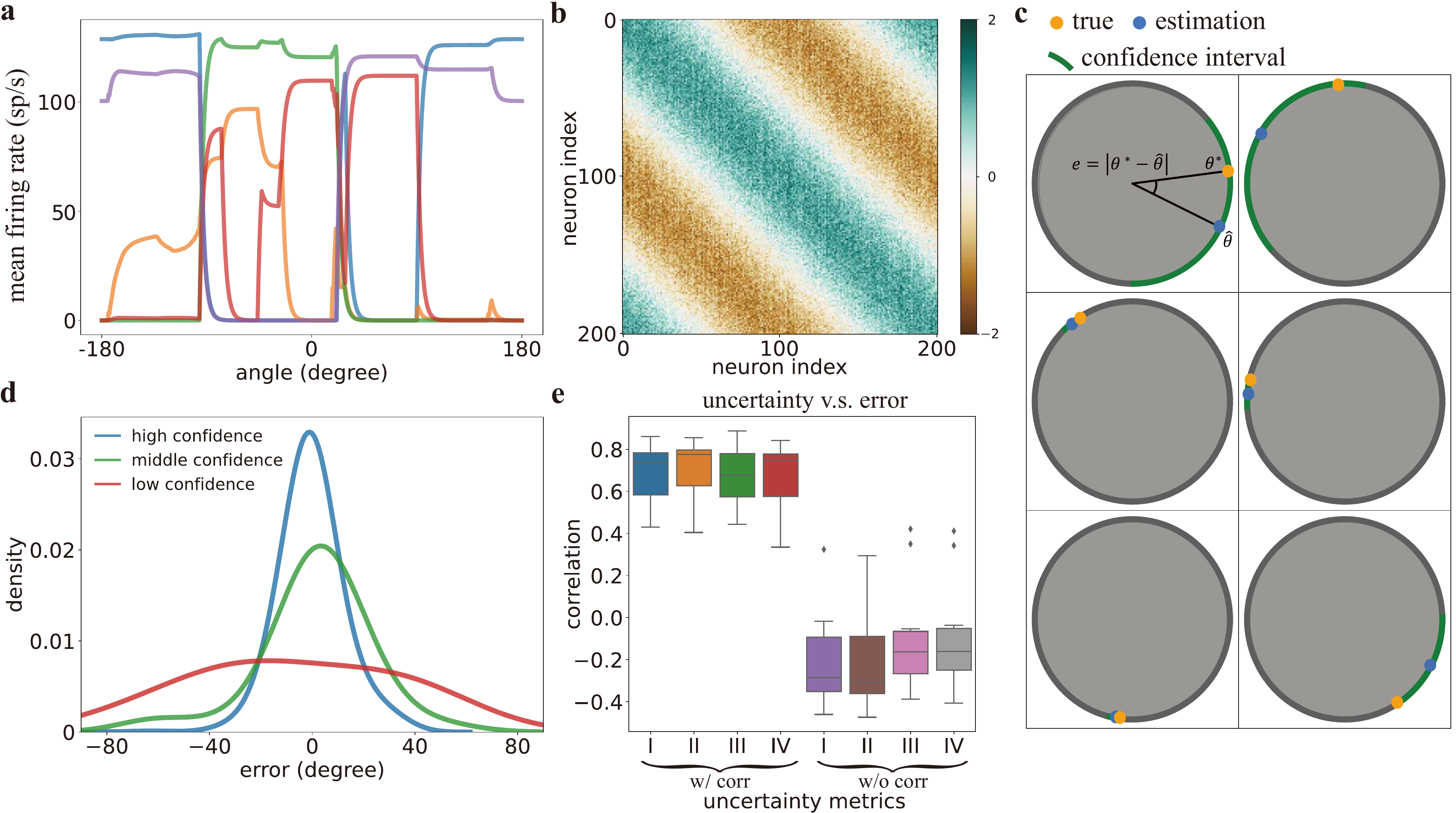}
    \caption{\textbf{The performance of moment neural network (MNN) trained on the working memory task.}  
    (a) The tuning curves of five neurons.  
    (b) The trained weights of the network.  
    (c) Six instances of the true input variables and the decoded feature and confidence interval decoded from the network. The arc length of the confidence interval is proportional to the square root of the first eigenvalue of the decoded covariance matrix $\hat{C}_{z}$, see {\bf Network inference} section in \mm.
    (d) The error distribution across three groups of instances, divided based on the level of uncertainty (using the uncertainty metrics I defined Eq.~\eqref{eq:uncern_1}, \mm): top 25\% uncertainty (low confidence), top 25-50\% uncertainty (middle confidence), and the remaining instances (high confidence). Corresponding results of uncertainty metrics II-IV are shown in the \si.
    (e) The correlation between uncertainty (calculated using four indicators, I-IV) and the error, calculated under two conditions: one where the correlation between neuron activities is maintained (w/ corr) and one where the correlation is clamped to zero (w/o corr).
 }
    \label{fig:uq_result}
\end{figure}

%% file: files/6_mechanism.tex
\section{Mechanism of the uncertainty quantification}

\input{fig/mechanism}

We analyzed the patterns of the fixed points, as shown in Fig.~\ref{fig:mechanism}a, where the mean and covariance of three fixed points are presented. Each pattern has been centered for comparison, and the five patterns are arranged in increasing order of covariance level.
The mean shows an imperfect bump, while the covariance displays a square-like pattern. More importantly, the bump width (defined as the number of neurons with firing rates exceeding 5\% of the peak) increases with the covariance level. To systematically verify this, we calculated the correlation coefficient between bump width and the four uncertainty indicators. As shown in Fig.~\ref{fig:mechanism}b, bump width is positively correlated with all four uncertainty indicators. In contrast, when firing correlations are clamped to zero, the correlation drops dramatically. This suggests that the uncertainty represented by the covariance is related to the bump width, which has been considered a form of probabilistic population coding for uncertainty~\cite{kutschireiter2023bayesian}.
In the probabilistic population coding framework, the mean encodes both the target variable and its associated uncertainty~\cite{ma2006bayesian, beck2011marginalization}. For instance, the bump amplitude or width~\cite{kutschireiter2023bayesian} represent uncertainty, where a smaller bump amplitude or larger bump width reflects greater uncertainty. In contrast, the sampling-based coding theory suggests that variability in neural activity encodes uncertainty, with higher variability corresponding to higher uncertainty. Our findings suggest that two theories capture distinct aspects of neural activity, and are linked through the nonlinear interaction between mean and covariance.

Next, we investigated how the nonlinear coupling between mean and covariance help to UQ. We propose that covariance reliably captures uncertainty through the following mechanism: when coding accuracy for an instance is low, the differential covariance component of the covariance increases, which has the effect of amplifying the level of uncertainty.
Differential covariance is the covariance component that changes the neural responses across trials along the direction
parallel to the tuning curve derivative hence shrinks its information it code~\cite{kohn2016correlations,panzeri2022structures,ma2023self}. As illustrated as Fig.~\ref{fig:mechanism}c, the differential covariance makes two codes of feature less distinguishable, reducing the coding accuracy, which has also been empirically verified in~\cite{ma2023self}.
If the coding accuracy of an instance is low, the bump location tends to drift from its original position, induces correlated neuronal fluctuation, as shown in the top of Fig.~\ref{fig:mechanism}d. We hypothesize that the covariance state accounts for the effects of these fluctuations, leading to the emergence of a differential covariance structure, where neurons on the same side of the bump are positively correlated, and neurons on opposite sides show negative correlation, as illustrated at the bottom of Fig.~\ref{fig:mechanism}d. Consequently, lower coding accuracy is associated with higher differential covariance components within the covariance. To verify this hypothesis, we introduce an indicator called the differential covariance ratio (DCR), which quantifies the amount of differential covariance (\mm). We calculated the correlation between DCR and the four uncertainty indicators. As shown in Fig.~\ref{fig:mechanism}e, DCR is highly positively correlated with all four uncertainty indicators, suggesting that differential covariance is the primary source of uncertainty, hence supporting our hypothesis.

Building on the intuitive interpretation above, we provide a rigorous proof using a simplified model to demonstrate how a neural system trained with a loss function that only supervises the mean can also learn the true output variance, emphasizing how the mean-covariance coupling contributes to effective uncertainty quantification. See Thm.~\ref{thm:consistency} and Thm.~\ref{thm:error_bound} in the \si.

%% file: fig/mechanism.tex
\begin{figure}
    \centering
    \includegraphics[width=0.99\linewidth]{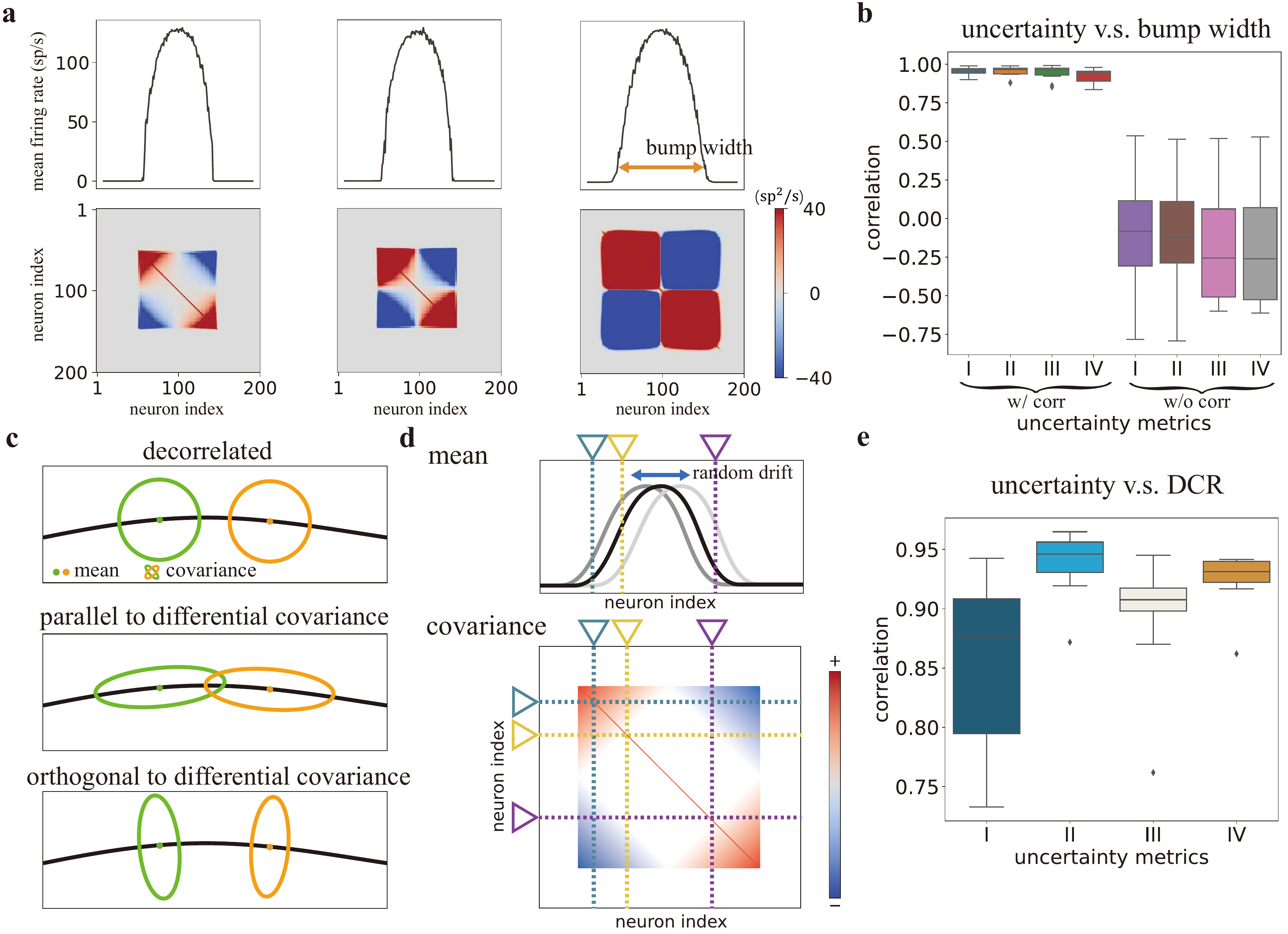}
    \caption{\textbf{Mechanism of uncertainty quantification in the moment neural network (MNN) for working memory.}  
(a) Several typical fixed-point patterns of mean firing rate (top) and firing covariance (bottom) produced by the MNN, with the bump width increasing from left to right.
(b) The correlation between the four uncertainty indicators (I-IV, Eq.~\eqref{eq:uncern_1}-Eq.~\eqref{eq:uncern_2}, \mm) and the bump width, calculated under two conditions: one where the correlation between neuron activities is maintained (w/ corr), and one where the correlation is clamped to zero (w/o corr).  
(c) A schematic showing how differential covariance affects the neural coding of variables with different values.  
(d) The random drift of the bump generates differential covariance components in the firing covariance.  
(e) The correlation between uncertainty and the differential covariance ratio (DCR).
    }
    \label{fig:mechanism}
\end{figure}

%% file: files/7_snn.tex
\input{fig/snn}

\section{Verification on spiking neural network}
We validated our findings on the SNN by transferring weights from the trained MNN to an SNN applying leaky integrate-and-fire neural model (Eq.~\eqref{eq:LIF}, \mm). As shown in Fig.~\ref{fig:snn}a, the tuning curves of the SNN were highly heterogeneous, consistent with those of the MNN. Again, each neuron responds to a narrow range of features, indicating a sparse coding scheme.
The steady-state mean firing rates and firing covariances of the SNN (Fig.~\ref{fig:snn}b) also exhibited qualitatively similar patterns to those observed in the MNN (Fig.~\ref{fig:mechanism}a). An analysis of the SNN's decoding results revealed that the error distribution (Fig.~\ref{fig:snn}c) and UQ performance (Fig.~\ref{fig:snn}d) were comparable to those of the MNN. These results confirm that the MNN effectively captures the mean and covariance dynamics of the SNN. Therefore, the MNN is a reliable and efficient alternative to directly training SNN for cognitive tasks. 

Furthermore, we observed a strong correlation between bump width and uncertainty in the SNN (Fig.~\ref{fig:snn}e), indicating the presence of mean-covariance coupling. Additionally, we calculated the correlation between uncertainty and the DCR in Fig.~\ref{fig:snn}f, which remained moderately positive, though lower than in the MNN (Fig.~\ref{fig:mechanism}e). This reduction may be attributed to fluctuations in the estimated covariance of the SNN, which likely diminish certain components of the differential covariance. These findings suggest that the proposed UQ mechanism and mean-covariance coupling identified in the MNN remain valid in the SNN.

We note several advantages of the MNN over the SNN:
First, different from the indifferentiable spike trains in the SNN, the moment activations of the MNN (Eq.~\eqref{eq:moment_activation}, \mm) are differentiable, enabling easier training.
Second, the MNN requires only a single pass to estimate the mean and covariance of the network state without fluctuations. In contrast, the SNN must be run multiple times to suppress fluctuations and achieve accurate estimates, leading to higher computational costs.
Third, the MNN allows investigations into the role of covariance by clamping the off-diagonal elements of the covariance matrix to zero. Performing such manipulations in the SNN is considerably more challenging. These advantages establish the MNN as a more flexible framework for investigating neural mechanisms underlying cognitive functions.

%% file: fig/snn.tex
\begin{figure}
    \centering
    \includegraphics[width=0.99\linewidth]{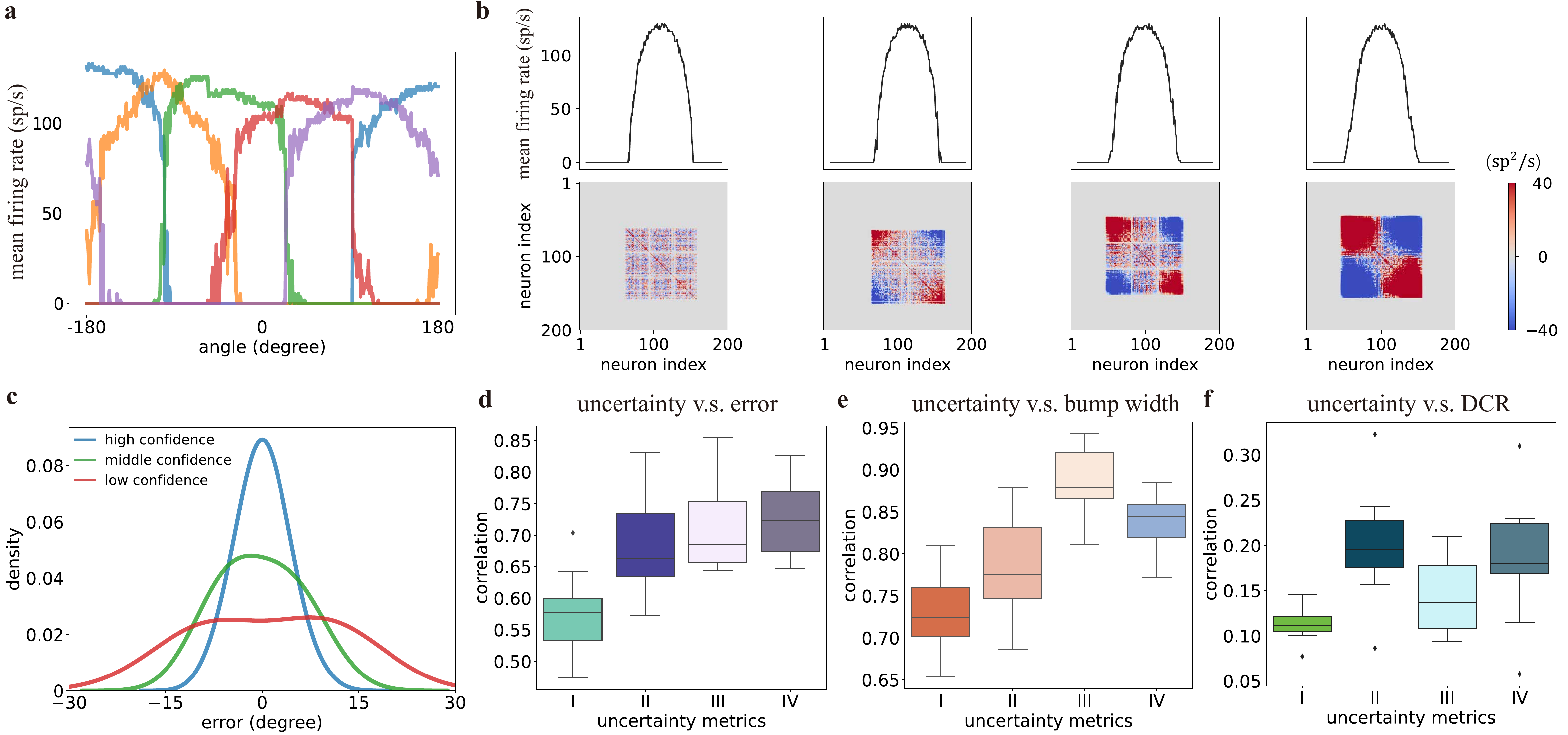}
    \caption{\textbf{The performance and spiking neural network (SNN) using the weights of the trained the moment neural network (MNN).}  
    (a) The tuning curves of five neurons.  
    The trained weights of the network.  
    (b) Several typical fixed-point patterns of mean firing rate (top) and firing covariance (bottom) produced by the SNN, with the bump width increasing from left to right.
    (c) The error distribution across three groups of instances, divided based on the level of uncertainty (using the uncertainty metrics I defined Eq.~\eqref{eq:uncern_1}, \mm): top 25\% uncertainty (low confidence), top 25-50\% uncertainty (middle confidence), and the remaining instances (high confidence). Corresponding results of uncertainty metrics II-IV are shown in the \si. 
    (d) The correlation between uncertainty (calculated using four indicators, I-IV) and the error.
    (e) The correlation between the four uncertainty indicators (I-IV, Eq.~\eqref{eq:uncern_1}-Eq.~\eqref{eq:uncern_2}, \mm) and the bump width.
    (f) The correlation between uncertainty and the differential covariance ratio (DCR). }
    \label{fig:snn}
\end{figure}

%% file: files/8_parameter.tex
\section{Further analysis}

\input{fig/parameter}

\input{fig/parameter2}

We examined how factors including noise level and heterogeneity influence the accuracy and UQ results of the MNN. This analysis not only enhances the plausibility of our model, but also provides a deeper understanding of how these factors shape the neural processes underlying WM, particularly in terms of whether noise and heterogeneity—two inevitable properties of biological neural systems—are detrimental or beneficial to WM performance. The following results applied the uncertainty metrics I (Eq.~\eqref{eq:uncern_1}, \mm), while the results of uncertainty metrics II-IV are shown in the \si.

As shown in Fig.~\ref{fig:parameter}a, as the number of training samples $M$ increases, the quality of UQ initially improves significantly, then levels off. Previous work has shown that $M = 12$ is sufficient for training a bump attractor network~\cite{darshan2022learning}, with fewer samples proving inadequate. This suggests that the network requires a sufficient number of samples to establish the underlying low-dimensional coding space of the feature $\theta$, with further increases in $M$ providing marginal improvement.
Additionally, the correlation between uncertainty and bump width is less affected by $M$, suggesting that the coupling between the mean and covariance is less influenced by network training and is instead an intrinsic property of the MNN.
Moreover, the error level continuously decreases as $M$ increases (Fig.~\ref{fig:parameter}c), which is consistent with the intuition that more training samples lead to better performance. 

We then investigated the effect of the population size $N$. As shown in Fig.~\ref{fig:parameter}d, $N$ has a similar effect on UQ performance as the number of training samples $M$, i.e., the network requires a sufficient number of neurons (about $200$) for stable UQ performance, with further increases in $N$ providing only marginal improvement. In contrast to $M$, increasing $N$ enhances the correlation between uncertainty and bump width (Fig.~\ref{fig:parameter}e). This suggests that enlarging the neuron population strengthens the coupling between mean and covariance. Additionally, the error level decreases as $N$ increases (Fig.~\ref{fig:parameter}f), likely because total information increases sub-linearly with population size, consistent with previous findings~\cite{averbeck2006neural}.

We analyzed the effect of the regularization factor $\alpha$ (Eq.~\eqref{eq:opt}, \mm). As shown in Figs.~\ref{fig:parameter}g and \ref{fig:parameter}i, there is an optimal value of $\alpha$ that yields the best UQ performance and minimal error level, respectively. This suggests that appropriate regularization is necessary for optimal network performance. Additionally, as shown in Fig.~\ref{fig:parameter}h, the regularization parameter $\alpha$ has less impact on the correlation between uncertainty and bump width, further supporting the idea that the coupling between mean and covariance is an intrinsic property of the MNN, largely unaffected by training.

Next, we analyzed the effect of noise. As shown in Figs.~\ref{fig:parameter2}a and \ref{fig:parameter2}c, the noise level during training, $\sigma_s$, significantly improves both UQ performance and accuracy. This suggests that noise plays a key role in error-awareness, underscoring its benefits for learning in the brain and highlighting its critical role in training neural networks for cognitive functions. The likely mechanism is that noise enhances the robustness of the network. Additionally, as shown in Fig.~\ref{fig:parameter2}b, $\sigma_s$ also increases the correlation between uncertainty and bump width, thereby strengthening the coupling between mean and covariance. As shown in Figs.~\ref{fig:parameter2}d-e, the noise level during inference, $\sigma_{s,2}$, increases the correlation between uncertainty and both error and bump width. Furthermore, increasing $\sigma_{s,2}$ raises the error level, as expected (Fig.~\ref{fig:parameter2}f). These results show that the MNN can maintain strong UQ performance across different levels of external noise, even when the training noise level is fixed. This suggests that the MNN generalizes its UQ ability to various noise conditions, a property that may be crucial for adaptation to new environments based on learned knowledge, as seen in humans and other species. We also provide a theoretic interpretation of such generalization ability in Thm.~\ref{thm:consistency} in the \si.

Finally, we investigated the effect of heterogeneity. As shown in Figs.~\ref{fig:parameter2}g-h, when heterogeneity $g$ is low, both the correlation between uncertainty and error, and the correlation between uncertainty and bump width are also low, resulting in poor UQ. This suggests that some degree of heterogeneity is necessary for the mean-covariance coupling. Interestingly, as shown in Fig.~\ref{fig:parameter2}i, the error level first increases and then decreases as heterogeneity increases. This result is similar to that observed in~\cite{kilpatrick2013optimizing}, where heterogeneity was introduced into a network to improve coding precision. The difference is that in our study, heterogeneity is unstructured, whereas in~\cite{kilpatrick2013optimizing} it is spatially structured. Thus, our results extend this earlier work by showing that the model can benefit from unstructured heterogeneity, which is more commonly observed in biological neural systems.

%% file: fig/parameter.tex
\begin{figure}
    \centering
    \includegraphics[width=0.99\linewidth]{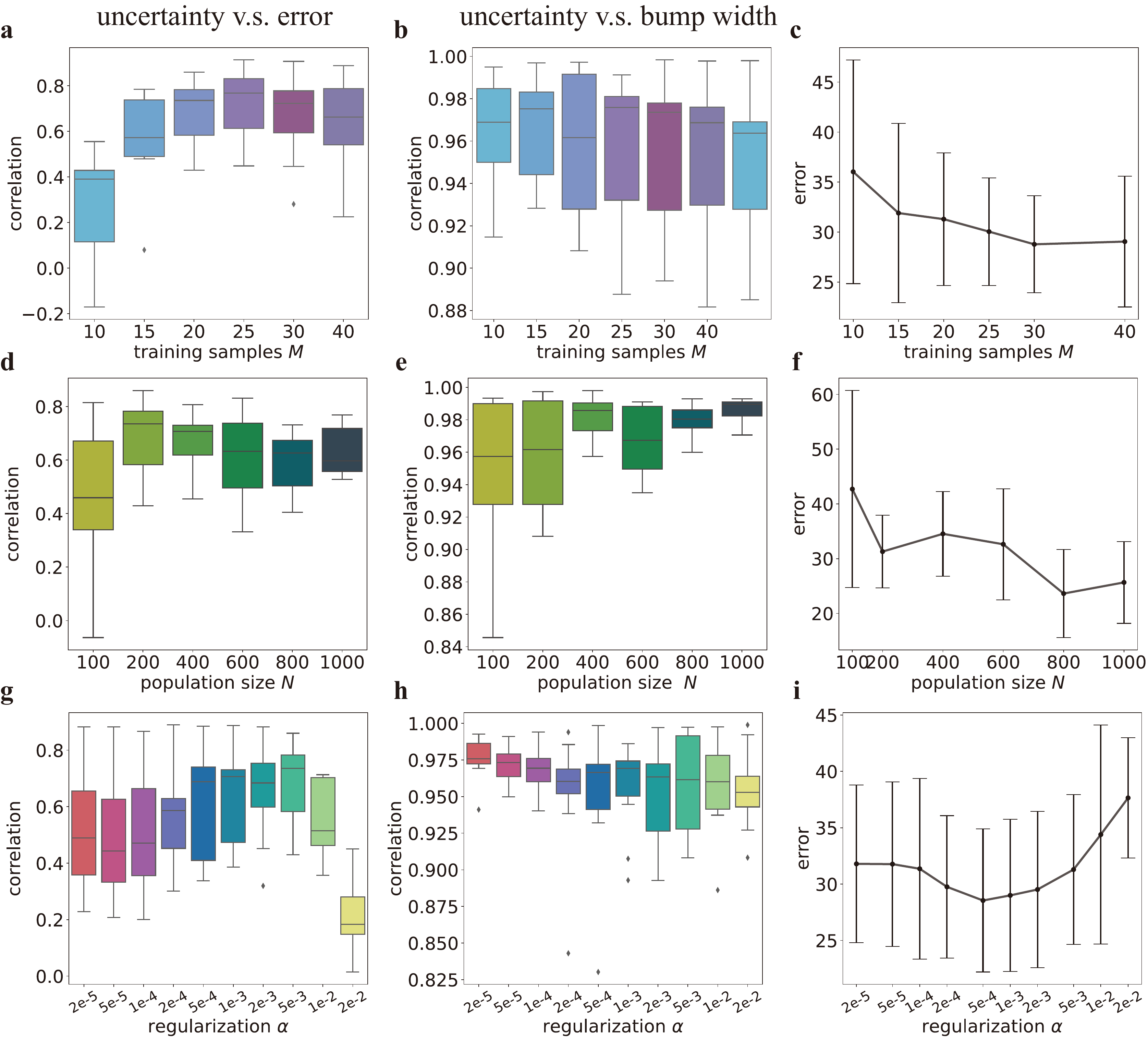}
    \caption{\textbf{Effect of training samples \(M\), population size \(N\), and regularization \(\alpha\) on uncertainty quantification, mean-covariance coupling, and precision in working memory (WM) tasks.}  
(a-c) We conduct the same experiments as in Fig.~\ref{fig:uq_result} with varying training samples \(M\), while keeping other parameters constant.  
(d-f) The same experiments as in (a-c), but with different population sizes \(N\).  
(g-i) The same experiments as in (a-c), but with different regularization values \(\alpha\).}
    \label{fig:parameter}
\end{figure}

%% file: fig/parameter2.tex
\begin{figure}
    \centering
    \includegraphics[width=0.99\linewidth]{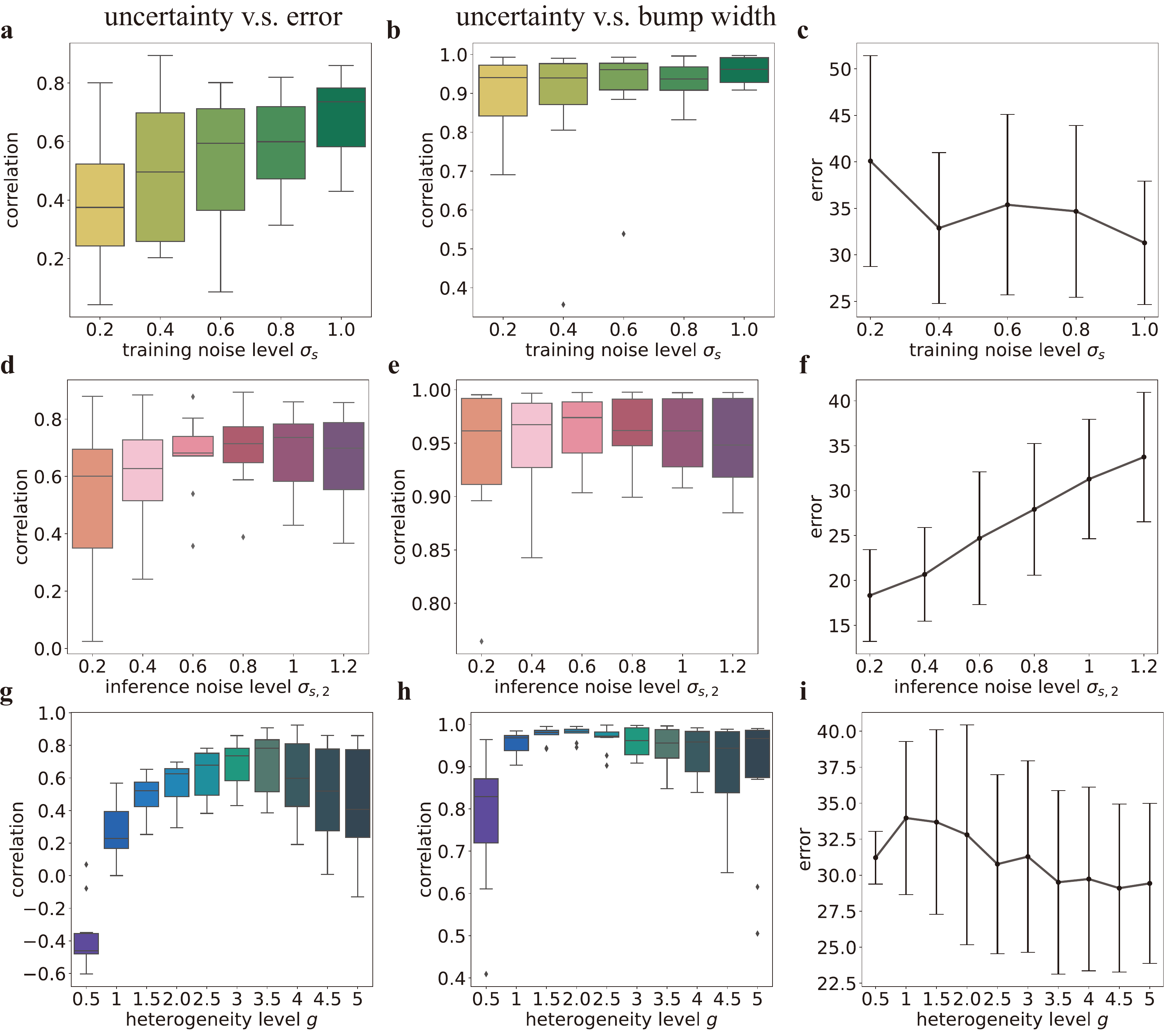}
  \caption{\textbf{Effect of level of noise at training phase \(\sigma_s\), level of noise at inference phase \(\sigma_{s,2}\), and heterogeneity \(g\) on uncertainty quantification, mean-covariance coupling, and precision in working memory (WM) tasks.} 
(a-c) We conduct the same experiments as in Fig.~\ref{fig:uq_result} with varying level of noise at training phase \(\sigma_s\), while keeping other parameters constant.  
(d-f) The same experiments as in (a-c), but with varying level of noise at inference phase \(\sigma_{s,2}\). (g-i) The same experiments as in (a-c), but with varying level of heterogeneity \(g\).}
    \label{fig:parameter2}
\end{figure}

%% file: files/9_discussion.tex
\section{Discussion}
Humans are prone to making mistakes in their working memory (WM), yet they can generate a sense of confidence that accurately reflects the level of error. In this study, we investigate the neural mechanisms underlying uncertainty quantification (UQ) using the moment neural network (MNN), which links probabilistic population coding and sampling-based coding theories through its nonlinear coupling of the mean and covariance. Trained with reservoir computing to encode a periodic feature, the MNN demonstrates UQ performance comparable to that of humans (Fig.~\ref{fig:uq_result}). These results are verified in an SNN using the trained MNN weights, resulting in a spike-based, heterogeneous, and sparse implementation of WM with its uncertainty quantification. Furthermore, we show that noise and heterogeneity, often seen as detrimental, actually enhance WM performance within a certain range (Figs.~\ref{fig:parameter2}d-i), suggesting they are essential for task performance rather than mere byproducts.

We establish a direct connection between probabilistic population coding and sampling-based coding theories, two prominent theories of how neuronal activity encodes uncertainty. As illustrated in Fig.~\ref{fig:mechanism}, we demonstrate that the bump width and the level of uncertainty calculated from the covariance are tightly coupled. The bump width serves as an indicator of uncertainty within the framework of probabilistic population coding, while the covariance acts as the uncertainty indicator in sampling-based coding. This coupling also explains why both theories are able to account for a wide range of experiments.
More importantly, we propose a hypothesis for how covariance faithfully captures uncertainty: the error arises from the random drift of the bump, which amplifies the differential covariance component within the covariance, thereby increasing the level of uncertainty. Supporting this hypothesis, we find that the level of covariance is strongly correlated with the amount of differential covariance (Fig.~\ref{fig:mechanism}e).

This study generates several testable predictions. First, the correlation between neural firing is crucial for both effective UQ (Fig.~\ref{fig:uq_result}e) and mean-covariance coupling (Fig.~\ref{fig:mechanism}b). 
This can be tested by decorrelating neuronal recording data to assess whether UQ performance is impaired and whether mean-covariance coupling disappears. Second, increasing the size of the neural population enhances UQ performance (Fig.~\ref{fig:parameter}d), although this effect may saturate once the population size exceeds a certain threshold. This prediction could be tested by comparing the number of neurons across different species and evaluating their respective quantification abilities. Third, the system may benefit from learning under noisy conditions (Fig.~\ref{fig:parameter2}a-c). This prediction could be tested by designing experiments where subjects learn a task under varying levels of noise, then assess how performance changes with different levels of training noise. Fourth, network heterogeneity may enhance both UQ and coding precision (Figs.~\ref{fig:parameter2}g-i). This hypothesis could be tested by quantifying heterogeneity across subjects and correlating it with performance.

The trade-off between model realism and feasibility for study and analysis is a significant challenge in neuroscience. Rate-based neural networks are widely used due to their ease of analysis and training; however, they sacrifice many biologically realistic features, such as spike communication and the intrinsic nonlinear dynamics of firing covariance, which play a crucial role in uncertainty quantification, as demonstrated in this work. On the other hand, spiking neural networks (SNNs) are more biologically accurate but much harder to analyze and train. The MNN applied in this work strikes a balance, combining the benefits of both approaches: the moment activations are differentiable, making the MNN easy to analyze and train like a rate-based neural network, while it captures the first two moments of the SNN, thus serving as a faithful substitute. Similar to our MNN, Echeveste et al. (2020) analyzed the dynamics of the first two moments of a neural network~\cite{echeveste2020cortical}. However, their work focuses on a rate-based neural network with additive noise, meaning that their network weights cannot be directly transferred to an SNN to replicate similar performance. In contrast, our verification experiments (Fig.~\ref{fig:snn}) demonstrate that by training an MNN via reservoir computing and transferring its weights to an SNN, we can achieve performance in WM tasks that is comparable to the MNN.

Our results shed light on how the brain learns for UQ. Notably, the MNN achieves high UQ performance despite its loss function (Eq.~\eqref{eq:opt}, \mm) does not explicitly incorporate firing covariance, the indicator of the uncertainty. This means that only the mean is directly supervised for the task, while the covariance is indirectly adjusted through the nonlinear coupling between the mean and covariance. This provides an example of how higher-order statistics of a system can be regulated through first-order statistics. These findings also offer valuable insights for designing UQ algorithms in machine learning, where UQ plays a critical role in applications such as active learning, reinforcement learning, domain adaptation, and security~\cite{gawlikowski2023survey}. Specifically, in a supervised learning setting, it suggests an approach where only the mean (prediction) of a model output is supervised, while the output covariance adapts unsupervised, allowing the covariance to learn to quantify the uncertainty associated with the mean.

Aitchison et al. (2021) suggest that uncertainty can be encoded in synaptic weights and regulate the synaptic adaptation rate~\cite{aitchison2021synaptic}. This implies that the uncertainty encoded in the MNN may regulate neuronal adaptation, playing a role in broader learning processes. Additionally, an intriguing direction for future research would be to explore how uncertainty is jointly or cooperatively encoded in the synaptic weights and neuron states and how they interact.

More broadly, the MNNs can be applied to model more complex cognitive processes. Since WM forms the foundation for higher-level functions, errors in WM can propagate to these processes, influencing the observable performance. Our study could be extended by analyzing how the firing covariance and the associated uncertainty in WM propagates to higher cognitive functions, offering insights into how humans quantify uncertainty in those tasks. In conclusion, our study provides opportunities for a deeper understanding of UQ and offers valuable insights into how biological systems manage uncertainty in noisy environments.

%% file: files/10_methods.tex
\section*{Methods}

{\bf Spiking neural networks (SNNs). }
 We applied leaky integrate-and-fire (LIF) neural model as for the SNN in the study. The dynamics of an SNN with $N$ LIF neurons is defined as
\begin{align}\label{eq:LIF}
\left\{\begin{matrix}
\frac{d\mathbf{u}}{dt} &=& -L\mathbf{u}+\mathbf{I},\\
\tau_s\frac{d\mathbf{I}}{dt} &=& -\mathbf{I}+W\mathbf{s},
\end{matrix}\right.
\end{align}
where $\mathbf{u}\in\mathbb{R}^N$ is the membrane potential each neuron, $L\in\mathbb{R}$ is the leak conductance, $W\in\mathbb{R}^{N\times N}$ is the synaptic weights, $\mathbf{I\in\mathbb{R}^N}$ is the synaptic currents, $\tau_s$ is the synaptic time constant, and $\mathbf{s}\in\mathbb{R}^N$ is the spike train generated by the neuron population. 
For \( i = 1, \ldots, N \), when \( u_i \) reaches the firing threshold \( u_{\text{th}} \), the neuron emits a spike (\( s_i = 1 \)) that is transmitted to connected neurons. Following the spike, \( u_i \) is reset to the resting potential \( u_{\text{res}} \) and enters a refractory period of duration \( T_{\text{ref}} \). If \( u_i \) does not reach \( u_{\text{th}} \), no spike is emitted (\( s_i = 0 \)).
Throughout this work, we set neuron parameters to be $u_{\rm th}=20$ mV, $u_{\rm res}=0$ mV, $T_{\rm ref}=5$ ms, $L=0.05$ ms$^{-1}$, and $\tau_s = 10$ ms.  

{\bf Moment neural networks (MNNs). }
We employed a neural model called as the MNN, which is derived from the LIF neural model~\cite{feng2006dynamics,lu2010gaussian}. Different from common rate-based neural model using elementiwise nonlinearity such as tanh and sigmoid.
The moment activations $\phi_{\mu}\in\mathbb{R}^{N}\times\mathbb{R}^{N\times N}\rightarrow\mathbb{R}^N$ and $\phi_{C}\in\mathbb{R}^{N}\times\mathbb{R}^{N\times N}\rightarrow\mathbb{R}^{N\times N}$ together map the mean $\bar{\bm{\mu}}$ and covariance $\bar{C}$ of the steady-state input current $\mathbf{I}=W\mathbf{s}$ to that of the output spikes according to~\cite{feng2006dynamics,lu2010gaussian}
\begin{align}\label{eq:moment_activation}
\begin{aligned}
    &\phi_{\mu}(\bar{\bm{\mu}},\bar{C})_i=  \big(T_{\rm ref} + \tfrac{2}{L}\int_{I_{\mathrm{ub},i}}^{I_{\mathrm{lb},i}} g(x) dx\big)^{-1},\\
    &\phi_{C}(\bar{\bm{\mu}},\bar{C})_{ij}  = 
    \left\{\begin{matrix} 
    &\frac{8}{L^2}\phi_{\mu}(\bar{\bm{\mu}},\bar{C})_{i}^{3}\textstyle\int_{I_{\mathrm{ub},i}}^{I_{\mathrm{lb},i}} h(x) dx,\quad &i=j\\
    &\big(\frac{\partial\phi_{\mu}(\bar{\bm{\mu}},\bar{C})}{\partial\bar{\bm{\mu}}}\big)_{ii}\big(\frac{\partial\phi_{\mu}(\bar{\bm{\mu}},\bar{C})}{\partial\bar{\bm{\mu}}}\big)_{jj}\bar{C}_{ij},\quad &i \neq j
\end{matrix}\right.,
\end{aligned}
\end{align}
where $g(x)$ and $h(x)$  are Dawson-like functions and defined as
\begin{align}
    g(x)=e^{x^2}\int_{-\infty}^x e^{-u^2}du, \quad
    h(x)=e^{x^2}\int_{-\infty}^x e^{-u^2}[g(u)]^2du.
\end{align}
and integration bounds are calculated as
\begin{align}
    I_{\mathrm{ub},i} = \frac{u_{\rm th}L-\bar{\mu}_i}{\sqrt{LC_{ii}}},\quad
    I_{\mathrm{lb},i} = \frac{u_{\rm res}L-\bar{\mu}_i}{\sqrt{LC_{ii}}}
\end{align}

{\bf Simulation of MNNs. }
All simulations were performed using the Euler method with a timestep \( dt = 0.1 \tau \) and \( \tau = 1 \) ms. The simulations were run on a single NVIDIA 3090 GPU. For the results presented in Fig.~\ref{fig:uq_result}-\ref{fig:parameter2}, each test was repeated 10 times. An efficient numerical algorithm was used to implement the moment neural network, as described in~\cite{PhysRevE.110.024310}.

{\bf Network connection and readout matrix. }
The connection weight matrix is supposed to have a low-rank structure with additive quenched noise that brings heterogeneity: 
\begin{align}
    W = W_{\mathrm{fb}}W_{\mathrm{out}}^\top + gJ,
\end{align}
where the random matrix $J_{ij}\sim \mathcal{N}(0,1/N)$ is fixed after being generated, and $g>0$ is the heterogeneity level. The matrix $W_{\mathrm{fb}}\in\mathbb{R}^{2\times N}$ is fixed as $W_{1j,fb} = \cos(2 \pi j/N), W_{2j,fb} = \sin(2 \pi j/N)$, following~\cite{darshan2022learning}, and the only trainable parameters are $W_{\mathrm{out}}\in\mathbb{R}^{2\times N}$, which is also the readout matrix.
We set the heterogeneity \( g=3 \) throughout the work, except in Fig.~\ref{fig:parameter2}, where we vary \( g \) from 0.5 to 5 to investigate its effect on the performance of the MNN. We set \( N=200 \) throughout the work, except in Fig.~\ref{fig:parameter}, where we vary the population size \( N \) from 100 to 1000 to investigate its effect on the performance of the MNN.

{\bf Encoding and decoding in the MNN. }
To input a cue of feature $\theta$ to the MNN, we encode it as the external current as $\bm{\mu}_s = W_{\mathrm{fb}}\mathbf{z}$ with
\begin{align}\label{eq:encode_z}
    \mathbf{z}_1 = A\cos(\theta),\quad \mathbf{z}_2 = A\sin(\theta),
\end{align}
where we set $A=1.2$ throughout the work.
When the state of the network as $(\bm{\mu},C)$, the decoded mean and covariance are
\begin{align}
    \hat{\bm{\mu}}_z= W_{\mathrm{out}}^\top \bm{\mu},\quad \hat{C}_{z}= W_{\mathrm
    {out}}^\top C W_{\mathrm{out}}.
\end{align}
The feature is decoded by finding the angle $\hat{\theta}$ satisfying
\begin{align}
    \cos(\hat{\theta})=\frac{\hat{\bm{\mu}}_{z,1}}{\sqrt{\hat{\bm{\mu}}_{z,1}^2+\hat{\bm{\mu}}_{z,2}^2}},\quad
    \sin(\hat{\theta})=\frac{\hat{\bm{\mu}}_{z,2}}{\sqrt{\hat{\bm{\mu}}_{z,1}^2+\hat{\bm{\mu}}_{z,2}^2}}.
\end{align}
To quantify the uncertainty associated with the decoded covariance $\hat{C}_{z}$ by considering the entropy of the Gaussian distribution $\mathcal{N}(\hat{\bm{\mu}}_z,\hat{C}_{z})$, which is a common indicators of uncertainty of a distribution. Please see the {\bf uncertainty metrics} section in \mm.

{\bf Network training. }
When uniformly select $M$ features from $[0,2\pi)$ denoted as $\theta^{(m)},m=1,\ldots,M$. For each feature $\theta^{(m)}$, we encode it as the external current $\bm{\mu}^{(m)}_s = W_{\mathrm{fb}}\mathbf{z}$ using Eq.~\eqref{eq:encode_z}, and set
\begin{align}
    &\bar{\bm{\mu}} = gJ\bm{\mu}+W_{\mathrm{fb}}\mathbf{z}+\bm{\mu}_s^{(m)}\\
    &\bar{C} = g^2JCJ^{\top}+\bm{\sigma}_s^2I
\end{align}
in Eq.~\eqref{eq:mnn}. We then simulate the network until its state convergence to a fixed point, denoted as $(\bm{\mu}^{(m)},C^{(m)})$.
After collecting the fixed points for every training feature $\bm{\mu}^{(m)}$, we train \( W_{\mathrm{out}} \) by minimizing the loss function
\begin{align}\label{eq:opt}
    \sum_{i=m}^M\left\|W_{\mathrm{out}}^\top \bm{\mu}^{(m)} - \mathbf{z}^{(m)}\right\|^2 + \alpha \left \| W_{\mathrm{out}} \right\|_2^2,
\end{align}
where $\alpha$ is the factor of regularization. Note that the covariance is not contained in the loss function. We set \( M=20 \), \( \sigma_s^2=1 \), and \( \alpha=5\times 10^{-3} \) throughout the work, except in Fig.~\ref{fig:parameter} and Fig.~\ref{fig:parameter2}, where we vary these parameters to investigate their respective effects on the performance of the MNN.

{\bf Network inference. }
During the cue period, to input a variable value $\theta^{\ast}$ for the model to hold as a memory item, we calculate the corresponding external input $\mathbf{z}^{\ast}$ and add external noise to the network by setting $\bm{\mu}_s$ as $W_{\mathrm{fb}}\mathbf{z}^\ast+\bm{\xi}_t$, where $\bm{\xi}_t\sim\mathcal{N}(\mathbf{0},\sigma_{\xi}^2I)$ is a Gaussian noise term, with $\sigma_{\xi}\in\mathrm{Unif}[0,1]$ independently sampled for each instance. The mean $\bar{\bm{\mu}}$ and covariance $\bar{C}$ of the synaptic inputs in Eq.~\eqref{eq:mnn} as
\begin{align}\label{eq:inference_input}
    &\bar{\bm{\mu}} = (gJ+W_{\mathrm{fb}}W_{\mathrm{out}}^\top)\bm{\mu}+\bm{\mu}_s\\
    &\bar{C} = (gJ+W_{\mathrm{fb}}W_{\mathrm{out}}^\top)C(gJ+W_{\mathrm{fb}}W_{\mathrm{out}}^\top)^{\top}+{\sigma}_{s,2}^2I,
\end{align}
where $\sigma_{s,2}$ is the noise level during the inference phase. The cue period lasts for $500$ time steps, after that we remove the external inputs by setting $\bm{\mu}_s=0$. Due the delay period, we simulate the network for $7500$ steps. After the delay period, we decode the remembered feature from the network.
The remembered variable is decoded from the mean as $\bm{\mu}_z=W_{\mathrm{out}}^\top\bm{\mu}$, and the corresponding covariance is calculated as $\hat{C}_{z}=W_{\mathrm{out}}^\top C W_{\mathrm{out}}$.
We calculate the error of the network as $e = |\hat{\theta}-\theta^{\ast}|$, where $\hat{\theta}$ and $\theta^{\ast}$ are the angles corresponds to $\hat{\bm{\mu}}_z$ and $\mathbf{z}^{\ast}$ respectively.

{\bf Uncertainty metrics. }
Based on the entropy of a Gaussian distribution, we design the following two uncertainty indicators
\begin{align}\label{eq:uncern_1}
    \text{(I): }\kappa = \mathrm{det}|\hat{C}_{z}|,\quad
    \text{(II): }\kappa = \log\mathrm{det}|\hat{C}_{z}|.
\end{align}
Denote the normal direction of the tangent space of $\mathcal{M}$ at $\hat{\bm{\mu}}_z$ as $n_{\hat{\bm{\mu}}_z}\in\mathbb{R}^2$, then we design another two uncertainty indicators
\begin{align}\label{eq:uncern_2}
    \text{(III): }\kappa = \mathrm{det}|n_{\hat{\bm{\mu}}_z}^\top\hat{C}_{z}n_{\hat{\bm{\mu}}_z}|,\quad
    \text{(IV): }\kappa = \log\mathrm{det}|n_{\hat{\bm{\mu}}_z}^\top\hat{C}_{z}n_{\hat{\bm{\mu}}_z}|.
\end{align}
These two uncertainty metrics measure the uncertainty along the tangent space of $\hat{\bm{\mu}}_z$ the task-related direction.

{\bf Correlation between the uncertainty metrics and the error. }
After trained a network, we select $L$ different angle samples $\theta_1^{\ast},\ldots,\theta_L^{\ast}$ and generate the corresponding external stimuli, which are exerted to the network separately. For each $\theta_i^{\ast}$, we calculate the error $e_i$ and the uncertainty $u_i$ ($u_i$ is calculated through one of I-IV indicators). We then estimate the correlation
\begin{align}\label{eq:rho}
    \rho = \frac{1}{L}\sum_i \kappa_i e_i - \frac{1}{L^2}\sum_i \kappa_i \sum_i e_i.
\end{align}
A large (positive) $\rho$ means that the uncertainty quantification is good. Throughout this work, we set $L=500$ and uniformly select $\theta_1^{\ast},\ldots,\theta_L^{\ast}$ from $[0,2\pi)$.

{\bf Differential covariance ratio (DCR). }We define the differential covariance ratio (DCR) as follows
\begin{align}
  \mathrm{Diff}(C) :=\frac{\sum_{ij}|C_{ij}|^2}{\sum_{ij}|C_{ij}|^2+\sum_{ij}|C_{ij}+C_{i,2N-j}+C_{i,2N-j}+C_{2N-i,2N-j}|^2} .
\end{align}
The DCR quantifies the differential covariance within the input covariance matrix \(C\), as the differential covariance induced by the drift of the bump exhibits a distinct pattern of sign relations (see Fig.~\ref{fig:mechanism})
\begin{align}
    \mathrm{sgn}(C_{ij}) =-\mathrm{sgn}(C_{2N-i,j}) = -\mathrm{sgn}(C_{i,2N-j}) = \mathrm{sgn}(C_{2N-i,2N-j})=1.
\end{align}
The higher the DCR, the larger the component of differential covariance in \(C\). The maximum value of the DCR is 1.

{\bf Simulation and decoding of the spiking neural network. }
For Fig.~\ref{fig:snn}, we simulated an SNN using the LIF neural model (Eq.~\eqref{eq:LIF}).
The durations of the cue and delay periods were set to match those used in the MNN simulation. To decode the feature and its uncertainty from the spike trains, we select a time interval \([T_{\mathrm{start}}, T_{\mathrm{end}}]\), and divided it into \(K = \tfrac{T_{\mathrm{end}} - T_{\mathrm{start}}}{\Delta T}\) time windows, where the window length \(\Delta T\) was chosen such that it evenly divided \(T_{\mathrm{end}} - T_{\mathrm{start}}\). For each time window indexed by \(k\) (\(k = 1, \ldots, K\)), we calculated the mean firing rate of each neuron within the window, denoted as \(\bar{\bm{r}}^{(k)}\in\mathbb{R}^N\). 
In this study, we set $T_{\mathrm{start}}=2000$, $T_{\mathrm{end}}=7500$ and$\Delta T=500$.
The mean firing rate and firing covariance of the SNN were then estimated as:
\[
\bm{\mu}_{\mathrm{snn}} = \frac{1}{K}\sum_{k=1}^K \bar{\bm{r}}^{(k)}, \quad
C_{\mathrm{snn}} = \frac{1}{K}\sum_{k=1}^K \left(\bar{\bm{r}}^{(k)} - \bm{\mu}_{\mathrm{snn}}\right)\left(\bar{\bm{r}}^{(k)} - \bm{\mu}_{\mathrm{snn}}\right)^\top.
\]
The subsequent decoding procedure was identical to that used for the MNN. For details, refer to the {\bf Encoding and decoding in the MNN} section.

%% file: files/Supplementary.tex
\begin{flushleft}
{\Large
\textbf{Supplementary Information:\\
Uncertainty Quantification in Working Memory via Moment Neural Networks
} %
}
\newline
{
  Hengyuan Ma\textsuperscript{1},
  Wenlian Lu\textsuperscript{1,2,3,4,5,6},
  Jianfeng Feng\textsuperscript{1,2,3,4,7$\ast$} 
\\

\bigskip
\it{1} Institute of Science and Technology for Brain-inspired Intelligence, Fudan University, Shanghai 200433, China
\\
\it{2} Key Laboratory of Computational Neuroscience and Brain-Inspired Intelligence (Fudan University), Ministry of Education, China
\\
\it{3} School of Mathematical Sciences, Fudan University, No. 220 Handan Road, Shanghai, 200433, Shanghai, China
\\
\it{4} Shanghai Center for Mathematical Sciences, No. 220 Handan Road, Shanghai, 200433, Shanghai, China
\\
\it{5} Shanghai Key Laboratory for Contemporary Applied Mathematics, No. 220 Handan Road, Shanghai, 200433, Shanghai, China
\\
\it{6} Key Laboratory of Mathematics for Nonlinear Science, No. 220 Handan Road, Shanghai, 200433, Shanghai, China 
\\
\it{7} Department of Computer Science, University of Warwick, Coventry, CV4 7AL, UK
\newline
$\ast$ jffeng@fudan.edu.cn 
}
\end{flushleft}

\setcounter{figure}{0}
\renewcommand{\thesection}{S\arabic{section}}
\renewcommand{\thefigure}{S\arabic{figure}}
\renewcommand{\thetable}{S\arabic{table}}
\renewcommand{\thetheorem}{S\arabic{theorem}}
\setcounter{section}{0}
\setcounter{equation}{0}
\renewcommand{\theequation}{S\arabic{equation}}
\setcounter{page}{1}

\section{Theorems and their proofs}
Here, a theoretical analysis of an abstract neural system is presented, demonstrating how the system can learn to output covariance that effectively quantifies uncertainty by training under a loss function that explicitly supervises only the mean of the network output.

We represent the neural system as a simplified parameterized function \( f(\mathbf{x}, \bm{\vartheta}) : \mathbb{R}^m \to \mathbb{R} \), where \(\bm{\vartheta}\) denotes the learnable parameters, and $\mathbf{x}$ is the input. Let \( f^{\ast}(\mathbf{x}) : \mathbb{R}^m \to \mathbb{R} \) represent the ground-truth neural system that the model aims to approximate. Internal and external noise are introduced to both systems: under specified noise conditions, the neuron’s inference result is computed as
\begin{align}
y(\mathbf{x},C_{\xi},\sigma^2_{\eta};\bm{\vartheta}) = f(\mathbf{x}+\bm{\xi}_1;\bm{\vartheta})+\eta_1,
\end{align}
and the ground-truth result is
\begin{align}
    y^{\ast}(\mathbf{x},C_{\xi},\sigma^2_{\eta}) = f^{\ast}(\mathbf{x}+\bm{\xi}_2)+\eta_2,
\end{align}
where $\bm{\xi}_1,\bm{\xi}_2\sim\mathcal{N}(\mathbf{0},C_{\xi})$ is the internal noise with covariance matrix $C_{\xi}\in\mathbb{R}^{m\times m}$, and $\eta_1,\eta_2\sim\mathcal{N}(0,\sigma^2_{\eta})$ is the external noise with the variance $\sigma^2_{\eta}$. We suppose that $C_{\xi}$ is strictly positive-definite, and $\sigma^2_{\eta}>0$.
 $\bm{\xi}_1,\bm{\xi}_2,\eta_1,\eta_2$ are independent from each other. The model output mean and output variance are calculated as
\begin{align}
    &m(\mathbf{x},C_{\xi},;\bm{\vartheta}) = \mathbb{E}[y(\mathbf{x},C_{\xi},\sigma^2_{\eta};\bm{\vartheta})], \\
    &v(\mathbf{x},C_{\xi},\sigma^2_{\eta};\bm{\vartheta}) = \mathbb{E}[\big(y(\mathbf{x},C_{\xi},\sigma^2_{\eta};\bm{\vartheta})-m(\mathbf{x},C_{\xi};\bm{\vartheta})\big)^2],
\end{align}
and the ground-truth mean and variance is
\begin{align}
    &m^{\ast}(\mathbf{x},C_{\xi}) = \mathbb{E}[y^{\ast}(\mathbf{x},C_{\xi},\sigma^2_{\eta})], \\
    &v^{\ast}(\mathbf{x},C_{\xi},\sigma^2_{\eta}) = \mathbb{E}[\big(y^{\ast}(\mathbf{x},C_{\xi},\sigma^2_{\eta})-m(\mathbf{x},C_{\xi})\big)^2].
\end{align}
Note that both \(m^{\ast}\) and \(m\) are not influenced by \(\sigma^2_{\eta}\) since \(\eta\) vanishes after taking the expectation, while both \(m^{\ast}\) and \(m\) are influenced by \(C_{\xi}\) due to the nonlinearity of \(f^{\ast}\) and \(f\), respectively. The neural system is trained to minimize the following loss, which only supervises the mean under the noise conditions \(C_{\xi}\) and \(\sigma^2_{\eta}\) during training 
\begin{align}
   \mathcal{L}(\bm{\vartheta},C_{\xi})= \int_{\mathbf{x}\in\mathbb{R}^m}\big(m^{\ast}(\mathbf{x},C_{\xi};\bm{\vartheta})-m(\mathbf{x},C_{\xi};\bm{\vartheta})\big)^2p(\mathbf{x})d\mathbf{x},
\end{align}
where $p(\mathbf{x})$ is the distribution of the input, which is assumed to be supported on the whole $\mathbb{R}^m$.

We first prove a theorem demonstrating that the neural system can learn the ground-truth output variance by minimizing the loss \(\mathcal{L}(\bm{\vartheta}, C_{\xi})\), which supervises only the mean output.
\begin{theorem}\label{thm:consistency}
    Supposed that the model parameter $\bm{\vartheta}$ diminishes the loss $\mathcal{L}(\bm{\vartheta},C_{\xi})$, then the model learns the ground-truth output mean and variance at the same time for any noise conditions with positive-definite ${C}'_{\xi}$, and ${\sigma}'^2_{\eta}$ 
    \begin{align}
       &m(\mathbf{x},{C}'_{\xi};\bm{\vartheta}) =  m^{\ast}(\mathbf{x},{C}'_{\xi}),\label{eq:mean_right}\\
      &v(\mathbf{x},{C}'_{\xi},{\sigma}'^2_{\eta};\bm{\vartheta}) =  v^{\ast}(\mathbf{x},{C}'_{\xi},{\sigma}'^2_{\eta}).\label{eq:cov_right}
    \end{align}
\end{theorem}
\begin{proof}
We note that the function $m(\mathbf{x},tC_{\xi};\bm{\vartheta})$ is the solution to the following parabolic differential equation
\begin{align}\label{eq:heat1}
    \left\{\begin{matrix}
    \partial_{t} w(t,\mathbf{x}) =\bigtriangledown_{\mathbf{x}} [C_{\bm{\xi}} \bigtriangledown_{\mathbf{x}} w(t,\mathbf{x})],&\quad t>0,&\quad  \mathbf{x}\in\mathbb{R}^n\\
 w(t,\mathbf{x}) = m(\mathbf{x},O;\bm{\vartheta}),&\quad t=0,&\quad \mathbf{x}\in\mathbb{R}^n\\
\end{matrix}\right.,
\end{align}
where $O\in\mathbb{R}^{m\times m}$ is the all-zero matrix.
Additionally, the function $m^{\ast}(\mathbf{x},tC_{\xi})$ is the solution to the following the same parabolic differential equation with a different initial condition
\begin{align}\label{eq:heat2}
    \left\{\begin{matrix}
    \partial_{t} w(t,\mathbf{x}) =\bigtriangledown_{\mathbf{x}} [C_{\bm{\xi}} \bigtriangledown_{\mathbf{x}} w(t,\mathbf{x})],&\quad t>0,&\quad  \mathbf{x}\in\mathbb{R}^n\\
 w(t,\mathbf{x}) = m^{\ast}(\mathbf{x},O),&\quad t=0,&\quad \mathbf{x}\in\mathbb{R}^n\\
\end{matrix}\right.,
\end{align}
Since $\mathcal{L}(\bm{\vartheta},C_{\xi},\sigma^2_{\eta})=0$, and $p(\mathbf{x})$ is support on the $\mathbb{R}^m$, we have
\begin{align}
    m(\mathbf{x},{C}_{\xi};\bm{\vartheta}) =  m^{\ast}(\mathbf{x},{C}_{\xi}),\quad \forall \mathbf{x}\in\mathbb{R}^m.
\end{align}
This means that the solutions to the two parabolic differential equations, Eq.~\eqref{eq:heat1} and Eq.~\eqref{eq:heat2}, are identical at \(t=1\). By applying the backward uniqueness property of parabolic differential equations~\cite{wu2019backward}, which states that if the solutions at a given time point \(t > 0\) are identical, then their initial states must also be identical, we have
\begin{align}
    m(\mathbf{x},O;\bm{\vartheta}) =  m^{\ast}(\mathbf{x},O),\quad \forall \mathbf{x}\in\mathbb{R}^m,
\end{align}
which is equivalent to 
\begin{align}
    f(\mathbf{x};\bm{\vartheta}) =  f^{\ast}(\mathbf{x}),\quad \forall \mathbf{x}\in\mathbb{R}^m.
\end{align}
Then both Eq.~\eqref{eq:mean_right} and Eq.~\eqref{eq:cov_right} hold for any ${C}'_{\xi},{\sigma}'^2_{\eta}$.
\end{proof}
Importantly, this theorem also suggests that the neural system inherently generalizes its ability to quantify uncertainty across different noise conditions, characterized by variations in \(C_{\xi}\) and \(\sigma^2_{\eta}\). This generalization ability has been observed in Figs.~\ref{fig:parameter2}d-f in the main paper

In general, minimizing the training loss to zero is impossible. However, we present a stronger theoretical result showing that the error in covariance can be effectively bounded by the error in the mean.
\begin{theorem}\label{thm:error_bound}
Supposed that both $f(\mathbf{x};\bm{\vartheta})$ and $f(\mathbf{x})$ are bounded by $B>0$.
Denote the loss under noise level $tI,\sigma^2_{\eta}$ as
\begin{align}
    \mathcal{L}(t;\bm{\vartheta}) = \mathcal{L}(\bm{\vartheta},tI,\sigma^2_{\eta}).
\end{align}
Define the error on the variance as
\begin{align}
\mathcal{L}_v(t;\bm{\vartheta})=\int_{\mathbf{x}}|v^{\ast}(\mathbf{x},tI,\sigma^2_{\eta}) - v(\mathbf{x},tI,\sigma^2_{\eta};\bm{\vartheta})|p(\mathbf{x})d\mathbf{x}
\end{align}
Then for any $\epsilon>0$, there exists $t_2>t_1>0$, such that
\begin{align}
    \mathcal{L}_v(t_1;\bm{\vartheta}) \leq 4B(\epsilon+\sqrt{2}\sigma_{\eta}+2\big(\mathcal{L}(t_2;\bm{\vartheta}) +  \frac{m}{2}\int_{t_1}^{t_2}\frac{ \mathcal{L}(t;\bm{\vartheta})}{t} dt\big)^{1/2}).
\end{align}
\end{theorem}
\begin{proof}
First, we apply
the parabolic differential equation (Eq.~\eqref{eq:heat1} and Eq.~\eqref{eq:heat2}) and the integration by parts, we have
\begin{align}
    \frac{d}{dt}\mathcal{L}(t;\bm{\vartheta}) = -2\int_{\mathbb{R}^n}\left\|\bigtriangledown_{\mathbf{x}} \big(m^{\ast}(\mathbf{x},C_{\xi};\bm{\vartheta})-y^{\ast}(\mathbf{x},C_{\xi};\bm{\vartheta})\big)\right\|^2p(\mathbf{x})d\mathbf{x}.
\end{align}
Hence we have for $0<t_1<t_2$
\begin{align}\label{eq:tau12}
   \mathcal{L}(t_1;\bm{\vartheta}) = \mathcal{L}(t_2;\bm{\vartheta})+ 2\int_{t_1}^{t_2} \int_{\mathbb{R}^n} \left\|\bigtriangledown_{\mathbf{x}} \big(m^{\ast}(\mathbf{x},C_{\xi};\bm{\vartheta})-y^{\ast}(\mathbf{x},C_{\xi};\bm{\vartheta})\big)\right\|^2 p(\mathbf{x})d\mathbf{x}d t.
\end{align}
Use the Harnack's inequality~\cite{evans2022partial}, we have
\begin{align}
    &\left \| \bigtriangledown_{\mathbf{x}} \big(m^{\ast}(\mathbf{x},C_{\xi};\bm{\vartheta})-y^{\ast}(\mathbf{x},C_{\xi};\bm{\vartheta})\big)\right\|^2 \leq r(\tau,\mathbf{x};\bm{\theta}) \partial_{\tau} \big(m^{\ast}(\mathbf{x},C_{\xi};\bm{\vartheta})-y^{\ast}(\mathbf{x},C_{\xi};\bm{\vartheta})\big)\\ 
    &+ \frac{m}{2\tau} \big(m^{\ast}(\mathbf{x},C_{\xi};\bm{\vartheta})-y^{\ast}(\mathbf{x},C_{\xi};\bm{\vartheta})\big)^2,
\end{align}
combine with Eq.~\eqref{eq:tau12}, we have
\begin{align}\label{eq:enegry}
  \mathcal{L}(t_1;\bm{\vartheta}) \leq \mathcal{L}(t_2;\bm{\vartheta}) +  \frac{m}{2}\int_{t_1}^{t_2}\frac{ \mathcal{L}(t;\bm{\vartheta})}{t} dt.
\end{align} 
Using the Hölder inequality and Minkowski inequality on the measure $p(\mathbf{x})$, we have
\begin{align}
\begin{aligned}
      \int_{\mathbf{x}}|y^{\ast}(\mathbf{x},\bm{\vartheta},O,\sigma_{\eta}^2)-y(\mathbf{x},\bm{\vartheta},O,\sigma_{\eta}^2)|p(\mathbf{x})d\mathbf{x} & \leq \big(\int_{\mathbf{x}}|y^{\ast}(\mathbf{x},\bm{\vartheta},O,\sigma_{\eta}^2)-y(\mathbf{x},\bm{\vartheta},O,\sigma_{\eta}^2)|^2p(\mathbf{x})d\mathbf{x}\big)^{1/2}\\
      &\leq 
        \big(\int_{\mathbb{R}^n}  (m(\mathbf{x},\bm{\vartheta},O)-m(\mathbf{x},\bm{\vartheta},t_1I))^2 p(\mathbf{x})d\mathbf{x}\big)^{1/2}\\
    &+
        \big(\int_{\mathbb{R}^n} (m(\mathbf{x},\bm{\vartheta},t_1I)-m^{\ast}(\mathbf{x},t_1I))^2 p(\mathbf{x})d\mathbf{x}\big)^{1/2}
    \\
     &+\big(\int_{\mathbb{R}^n}  (m^{\ast}(\mathbf{x},t_1I)-m^{\ast}(\mathbf{x},\mathbf{x},O))^2 d\mathbf{x}\big)^{1/2}+\sqrt{2}\sigma_{\eta}\\
     &=\big(\int_{\mathbb{R}^n}  (m(\mathbf{x},\bm{\vartheta},O)-m(\mathbf{x},\bm{\vartheta},t_1I))^2 p(\mathbf{x})d\mathbf{x}\big)^{1/2}+\\
     &+\big(\int_{\mathbb{R}^n}  (m^{\ast}(\mathbf{x},t_1I,\sigma_{\eta}^2)-m^{\ast}(\mathbf{x},O))^2 d\mathbf{x}\big)^{1/2}+    
       \mathcal{L}(t_1;\bm{\vartheta})^{1/2}+\sqrt{2}\sigma_{\eta}
\end{aligned},
\end{align}
where $O\in\mathbb{R}^{m\times m}$ is the all-zero matrix.
Noticed that for $a,b,{a}',{b}'\in\mathbb{R}$, we have
\begin{align}
    |(a-b)^2 - ({a}'-{b}')^2| \leq |a+{a}'-2b||a-{a}'|+|2{a}'-b-{b}'||b-{b}'| \leq 4\max\{|a|,|{a}'|,|b|,|{b}'|\}(|a-{a}'|+|b-{b}'|).
\end{align}
we have
\begin{align}
    &\int_{\mathbf{x}}|v^{\ast}(\mathbf{x},t_1I,\sigma^2_{\eta}) - v(\mathbf{x},t_1I,\sigma^2_{\eta};\bm{\vartheta})|p(\mathbf{x})d\mathbf{x} \\
    &\leq 4B\left(\int_{\mathbf{x}}\mathbb{E}[|y^{\ast}(\mathbf{x},t_1I,\sigma^2_{\eta})-y(\mathbf{x},t_1I,\sigma^2_{\eta})|]p(\mathbf{x})d\mathbf{x}+
    \int_{\mathbf{x}}|m^{\ast}(\mathbf{x},t_1I,\sigma^2_{\eta})-m(\mathbf{x},t_1I,\sigma^2_{\eta})|p(\mathbf{x})d\mathbf{x}\right)\\
    &\leq 4B\left(\left(\int_{\mathbf{x}}\mathbb{E}[\big(y^{\ast}(\mathbf{x},t_1I,\sigma^2_{\eta})-y(\mathbf{x},t_1I,\sigma^2_{\eta})\big)^2]p(\mathbf{x})d\mathbf{x}\right)^{1/2}+
    \mathcal{L}(t_1;\bm{\vartheta})^{1/2}\right)\\
   &= 4B\left(\left(\int_{\mathbf{x}}\big(y^{\ast}(\mathbf{x},O,\sigma^2_{\eta})-y(\mathbf{x},O,\sigma^2_{\eta})\big)^2\tilde{p}(\mathbf{x})d\mathbf{x}\right)^{1/2}+
    \mathcal{L}(t_1;\bm{\vartheta})^{1/2}\right),
\end{align}
where $\tilde{p}(\mathbf{x})$ is the distribution of $\mathbf{x}+\bm{\xi}$ with $\mathbf{x}\sim p(\mathbf{x})$ and $\bm{\xi}\sim\mathcal{N}(\mathbf{0},t_1I)$.
Due the continuity of the Gaussian convolution kernel and the assumption that $y$ and $y^\ast$ are bounded, for a given $\epsilon>0$, there exists $t_1>0$ small enough, such that 
\begin{align}
    &\big(\int_{\mathbb{R}^n}  (m(\mathbf{x},\bm{\vartheta},O)-m(\mathbf{x},\bm{\vartheta},t_1I))^2 p(\mathbf{x})d\mathbf{x}\big)^{1/2}
    +\big(\int_{\mathbb{R}^n}  (m^{\ast}(\mathbf{x},t_1I,\sigma_{\eta}^2)-m^{\ast}(\mathbf{x},O))^2 d\mathbf{x}\big)^{1/2}\\
    &+\left|\left(\int_{\mathbf{x}}\big(y^{\ast}(\mathbf{x},O,\sigma^2_{\eta})-y(\mathbf{x},O,\sigma^2_{\eta})\big)^2\tilde{p}(\mathbf{x})d\mathbf{x}\right)^{1/2} - \left(\int_{\mathbf{x}}\big(y^{\ast}(\mathbf{x},O,\sigma^2_{\eta})-y(\mathbf{x},O,\sigma^2_{\eta})\big)^2p(\mathbf{x})d\mathbf{x}\right)^{1/2}\right| \leq \epsilon.
\end{align}
Then we have
\begin{align}
    \int_{\mathbf{x}}|v^{\ast}(\mathbf{x},t_1I,\sigma^2_{\eta}) - v(\mathbf{x},t_1I,\sigma^2_{\eta};\bm{\vartheta})|p(\mathbf{x})d\mathbf{x}\leq 4B(\epsilon+\sqrt{2}\sigma_{\eta}+2\mathcal{L}(t_1;\bm{\vartheta})^{1/2}).
\end{align}
Combine with Eq.~\eqref{eq:enegry}, we have
\begin{align}
      \int_{\mathbf{x}}|v^{\ast}(\mathbf{x},t_1I,\sigma^2_{\eta}) - v(\mathbf{x},t_1I,\sigma^2_{\eta};\bm{\vartheta})|p(\mathbf{x})d\mathbf{x}\leq 4B(\epsilon+\sqrt{2}\sigma_{\eta}+2\big(\mathcal{L}(t_2;\bm{\vartheta}) +  \frac{m}{2}\int_{t_1}^{t_2}\frac{ \mathcal{L}(t;\bm{\vartheta})}{t} dt\big)^{1/2}),
\end{align}
which proves the theorem.
\end{proof}
This theorem states that the error in the output variance can be controlled by the mean loss under higher noise conditions (\(t_2 > t_1\)). This suggests that the neural system should be trained under a higher noise level to effectively quantify uncertainty within the noise range present during inference.

Both Thm.~\ref{thm:consistency} and Thm.~\ref{thm:error_bound} can be extended to cases where the output \(y\) is high-dimensional.

\section{Supplementary figures}

\begin{figure}[h]
    \centering
    \includegraphics[width=0.95\linewidth]{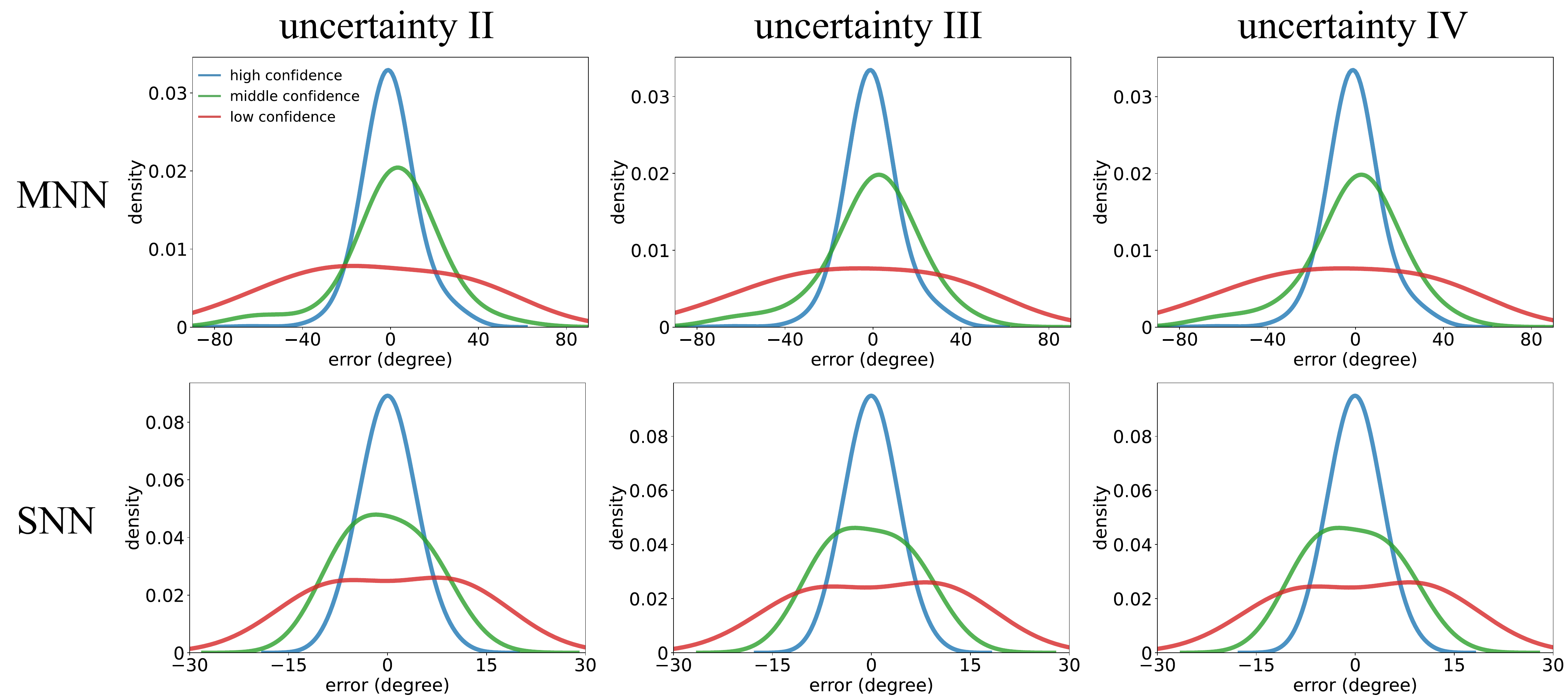}
    \caption{The error distribution across three groups of instances, divided based on the level of uncertainty (using the uncertainty metrics II, III, and IV (Eq.~\eqref{eq:uncern_1}-\eqref{eq:uncern_2}, \mm)): top 25\% uncertainty (low confidence), top 25-50\% uncertainty (middle confidence), and the remaining instances (high confidence). We plot the result for both moment neural network (MNN) and spiking neural network (SNN). The corresponding results calculated by the uncertainty metrics I are shown in Fig.~\ref{fig:uq_result}d and Fig.~\ref{fig:snn}c, respectively.
    }
    \label{fig:SI_error_dist}
\end{figure}

\begin{figure}[h]
    \centering
    \includegraphics[width=0.95\linewidth]{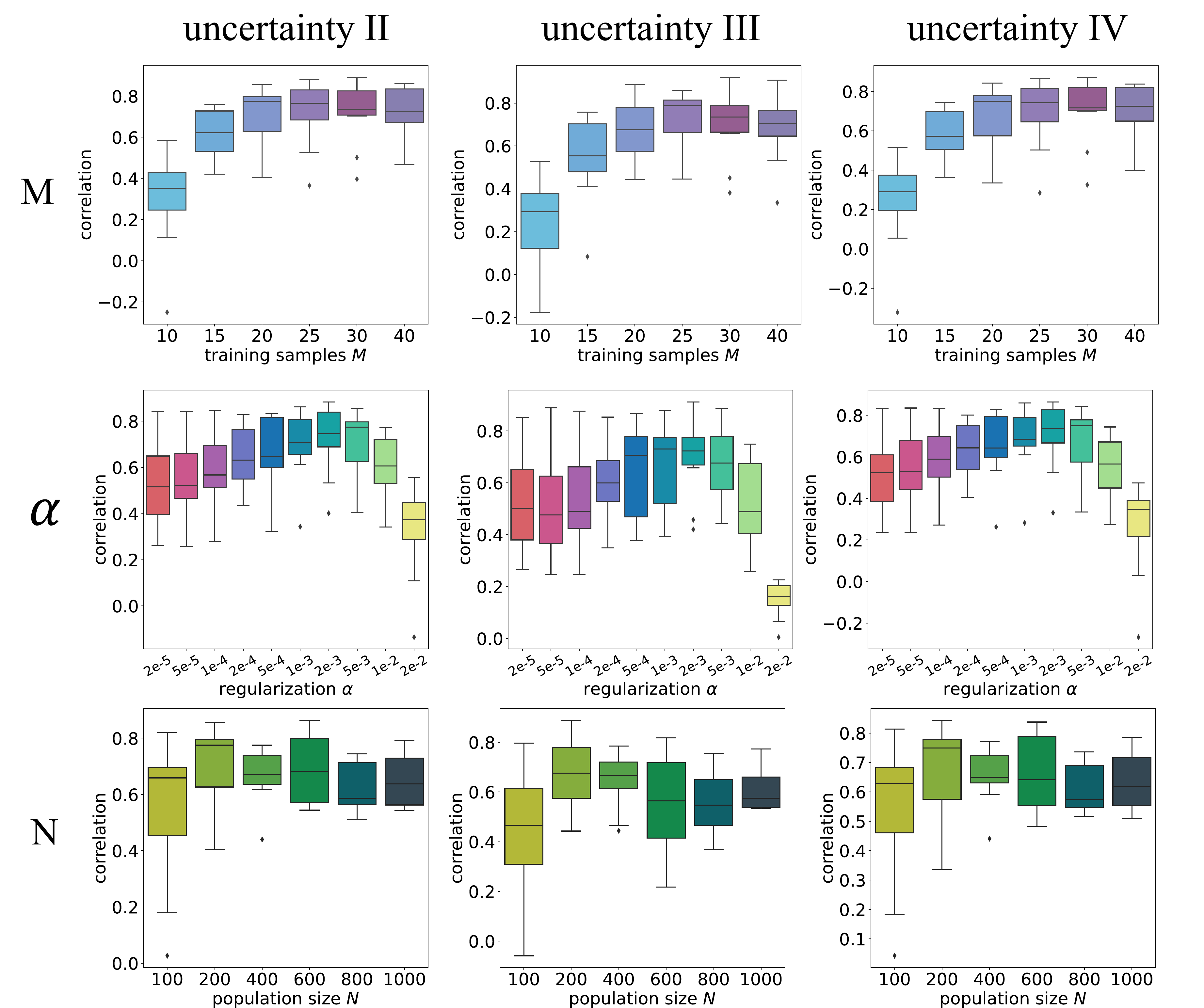}
    \caption{Effect of training samples \(M\), population size \(N\), and regularization \(\alpha\) on the uncertainty quantification performance (evaluated by the correlation between the error and the uncertainty quantification calculated by the uncertainty metrics (Eq.~\eqref{eq:uncern_1}-\eqref{eq:uncern_2}, \mm). The corresponding results of uncertainty metrics I is shown in Fig.~\ref{fig:parameter}a,d,g.
    }
    \label{fig:SI_parameter_corr}
\end{figure}

\begin{figure}[h]
    \centering
    \includegraphics[width=0.95\linewidth]{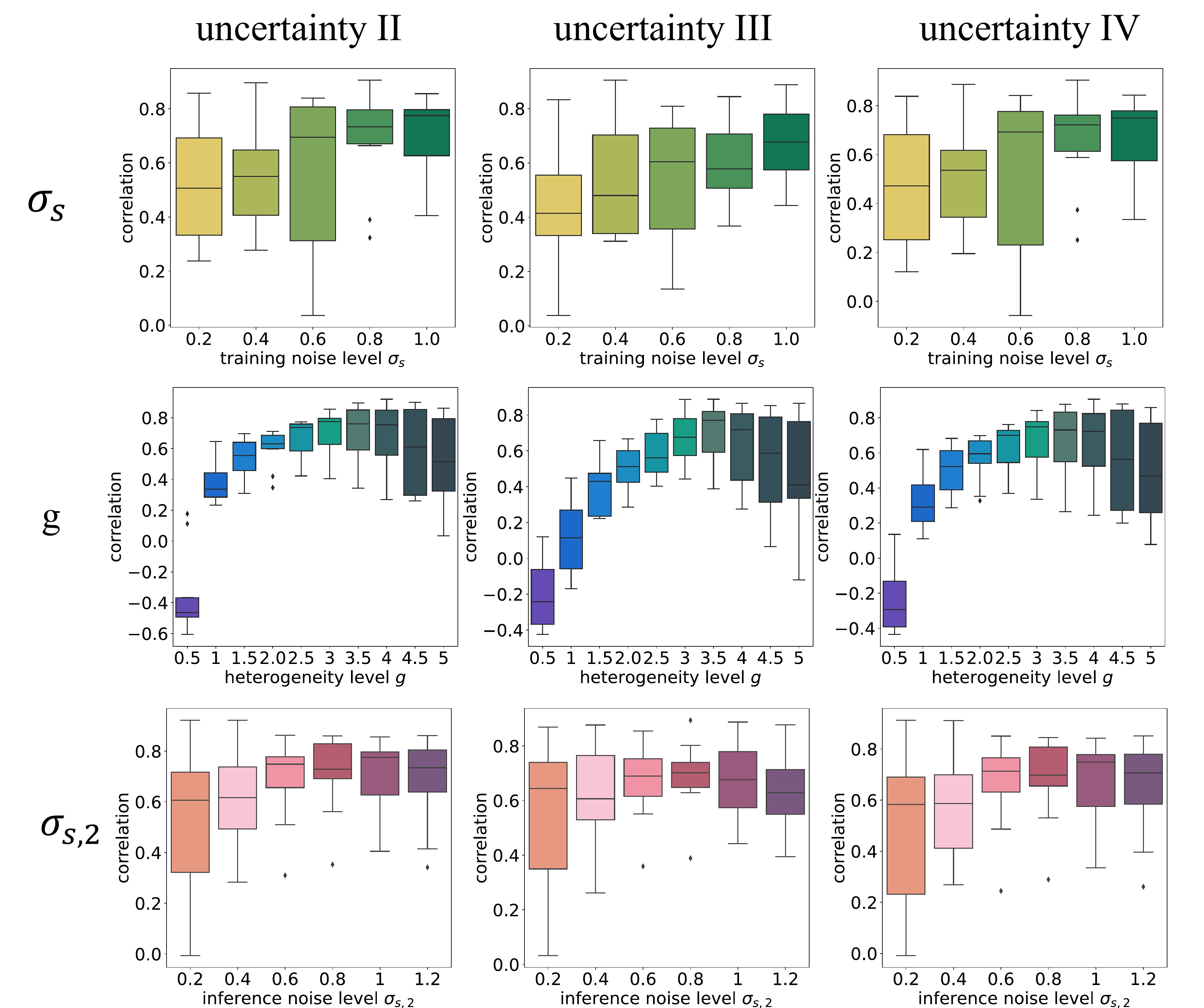}
    \caption{Effect of level of noise at training phase \(\sigma_s\), heterogeneity \(g\) and level of noise at inference phase \(\sigma_{s,2}\)  on the uncertainty quantification performance (evaluated by the correlation between the error and the uncertainty quantification calculated by the uncertainty metrics (Eq.~\eqref{eq:uncern_1}-\eqref{eq:uncern_2}, \mm)) The corresponding results of uncertainty metrics I is shown in Fig.~\ref{fig:parameter2}a,d,g.
    }
    \label{fig:SI_parameter2_corr}
\end{figure}

\begin{figure}[h]
    \centering
    \includegraphics[width=0.95\linewidth]{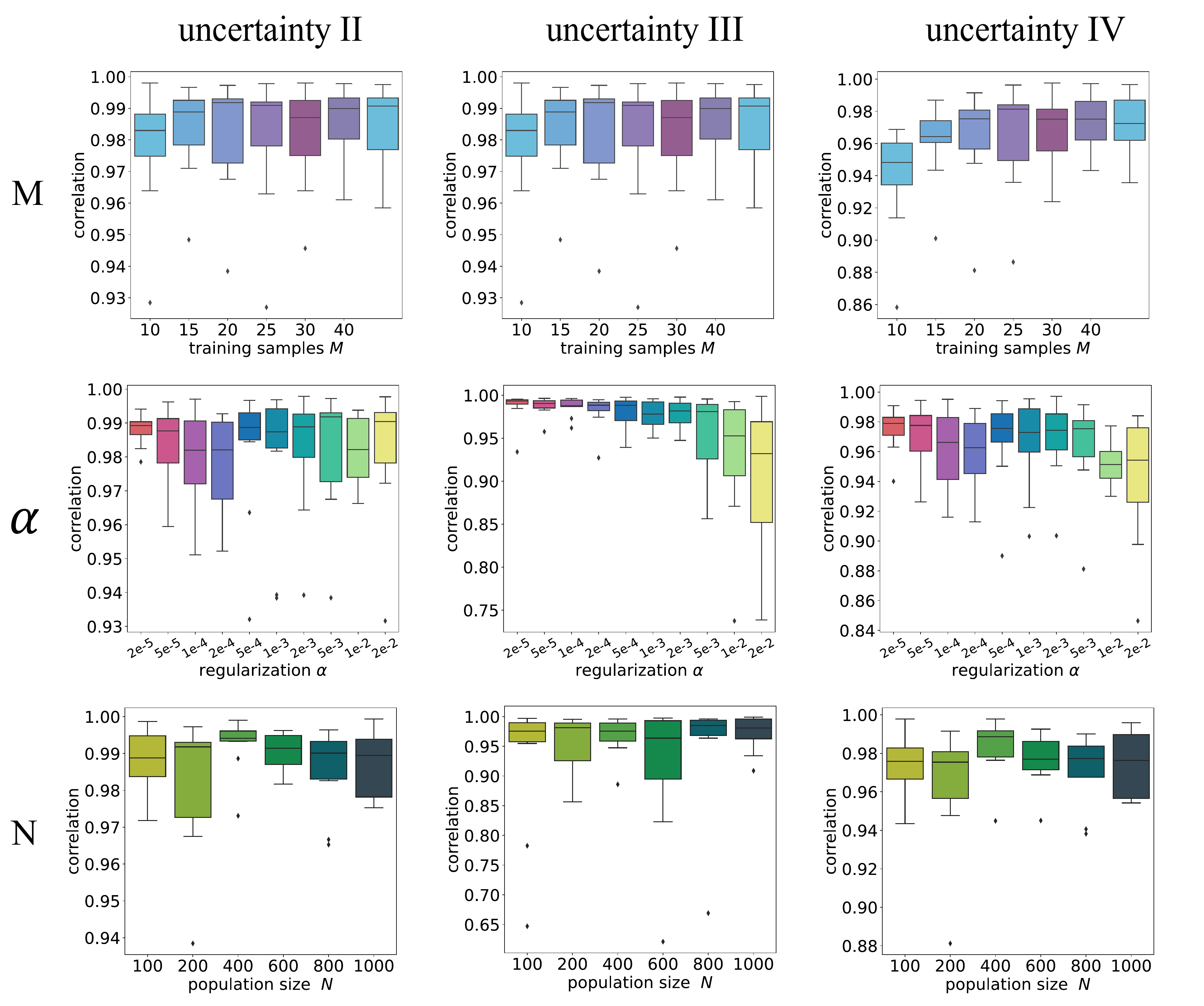}
    \caption{Effect of training samples \(M\), population size \(N\), and regularization \(\alpha\) on the mean-covariance coupling strength (evaluated by the correlation between the bump width and the uncertainty quantification calculated by the uncertainty metrics (Eq.~\eqref{eq:uncern_1}-\eqref{eq:uncern_2}, \mm). The corresponding results of uncertainty metrics I is shown in Fig.~\ref{fig:parameter}b,e,h.
    }
    \label{fig:SI_parameter_width}
\end{figure}

\begin{figure}[h]
    \centering
    \includegraphics[width=0.95\linewidth]{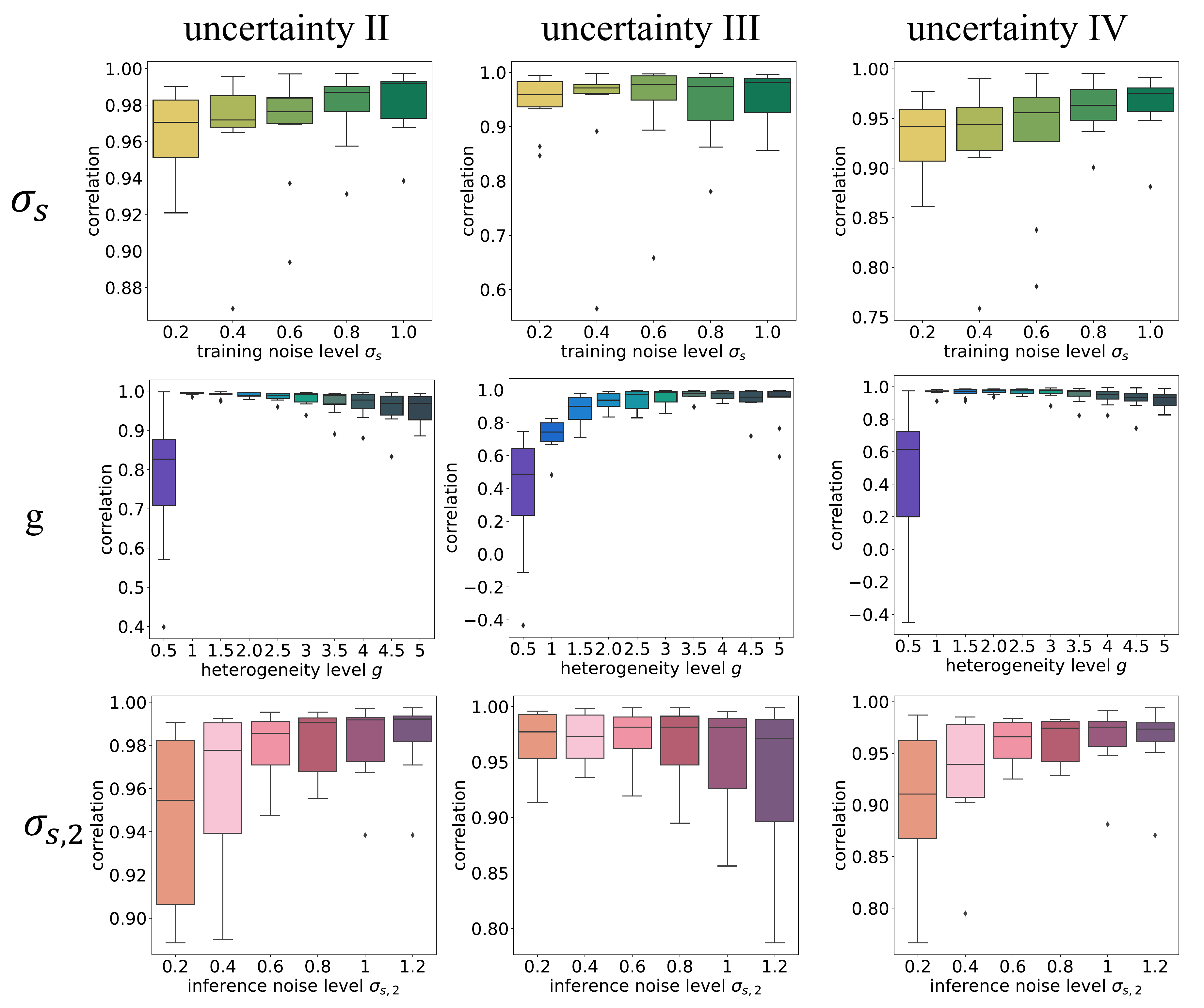}
    \caption{Effect of level of noise at training phase \(\sigma_s\), heterogeneity \(g\) and level of noise at inference phase \(\sigma_{s,2}\)  on the mean-covariance coupling strength (evaluated by the correlation between the bump width and the uncertainty quantification calculated by the uncertainty metrics (Eq.~\eqref{eq:uncern_1}-\eqref{eq:uncern_2}, \mm). The corresponding results of uncertainty metrics I is shown in Fig.~\ref{fig:parameter2}b,e,h.
    }
    \label{fig:SI_parameter2_width}
\end{figure}